\documentclass[a4paper, 11pt]{article}
\usepackage[margin=3.75cm]{geometry}

\usepackage[utf8]{inputenc}

\usepackage{color}

\definecolor{red}{rgb}{1.00,0.35,0.35}      
\definecolor{green}{rgb}{0.1,0.75,0.075}    
\definecolor{blue}{rgb}{0.35,0.45,1.00}     
\definecolor{darkblue}{rgb}{0.11,0.17,0.55} 
\definecolor{cyan}{rgb}{0.50,0.85,1.00}     
\definecolor{pink}{rgb}{1.0,0.35,1.0}       
\definecolor{orange}{RGB}{255, 166, 61}
\definecolor{purple}{RGB}{225, 103, 255}
\definecolor{lightgreen}{rgb}{0.60,0.90,0.60}

\usepackage{graphicx}
\usepackage{tikz}
	\usetikzlibrary{shapes.misc}
\tikzset{inner sep = 2pt,
	gnode/.style={
		radius=4pt,
		circle,
		minimum size=4pt
	},
	edgeto/.style={
		-latex,
		shorten <=2pt,
		shorten >=2pt,
	},
	edgeboth/.style={
		latex-latex,
		shorten <=2pt,
		shorten >=2pt,
	},
	edgefrom/.style={
		latex-,
		shorten <=2pt,
		shorten >=2pt,
	},
	pictikz-thick/.style={
		very thick
	},
	pictikz-dashed/.style={
		thick,
		dashed,
	},
	pictikz-dotted/.style={
		thick,
		dotted,
	}
}
\definecolor{black}{RGB}{0, 0, 0}
\definecolor{white}{RGB}{255, 255, 255}
\definecolor{gray}{RGB}{128, 128, 128}
\definecolor{red}{RGB}{255, 83, 83}
\definecolor{green}{RGB}{101, 237, 101}
\definecolor{blue}{RGB}{123, 168, 255}
\definecolor{yellow}{RGB}{246, 246, 82}
\definecolor{cyan}{RGB}{87, 239, 239}
\definecolor{purple}{RGB}{225, 103, 255}
\definecolor{light-green}{RGB}{178, 255, 66}
\definecolor{orange}{RGB}{255, 166, 61}
\definecolor{pink}{RGB}{255, 141, 255}

\usetikzlibrary{shapes.misc}

\tikzset{
	pictikz-node/.style={
		inner sep=1.0pt,
		circle,
		radius=3pt,
		minimum size=6pt
	},
	pictikz-rectangle/.style={
		rounded rectangle,
		inner sep=3pt,
		minimum width=6pt
	},
	pictikz-edgeto/.style={
		-latex,
		shorten <=2pt,
		shorten >=2pt,
	},
	pictikz-edgeboth/.style={
		latex-latex,
		shorten <=2pt,
		shorten >=2pt,
	},
	pictikz-edgefrom/.style={
		latex-,
		shorten <=2pt,
		shorten >=2pt,
	},
	pictikz-thick/.style={
		very thick
	},
	pictikz-dashed/.style={
		thick,
		dashed,
	},
	pictikz-dotted/.style={
		thick,
		dotted,
	}
}
\definecolor{pictikz-black}{RGB}{0, 0, 0}
\definecolor{pictikz-white}{RGB}{255, 255, 255}
\definecolor{pictikz-gray}{RGB}{128, 128, 128}
\definecolor{pictikz-red}{RGB}{255, 83, 83}
\definecolor{pictikz-green}{RGB}{101, 237, 101}
\definecolor{pictikz-blue}{RGB}{123, 168, 255}
\definecolor{pictikz-yellow}{RGB}{246, 246, 82}
\definecolor{pictikz-cyan}{RGB}{87, 239, 239}
\definecolor{pictikz-purple}{RGB}{225, 103, 255}
\definecolor{pictikz-light-green}{RGB}{178, 255, 66}
\definecolor{pictikz-orange}{RGB}{255, 166, 61}
\definecolor{pictikz-pink}{RGB}{255, 141, 255}

\usepackage{subcaption}

\usepackage{algorithm}

\usepackage{amsmath}
\usepackage{amsfonts}
\usepackage{amssymb}
\usepackage{amsthm}
\usepackage{mathtools}
\theoremstyle{plain}
\newtheorem{theorem}{Theorem}[section]
\newtheorem{lemma}{Lemma}

\newtheorem{observation}{Observation}

\newtheorem{rrule}{Reduction Rule}[section]
\newtheorem{rrule*}{Reduction Rule}

\theoremstyle{definition}

\theoremstyle{remark}

\newtheorem{definition}[theorem]{Definition}

\usepackage{dsfont}

\usepackage{float}
\usepackage{booktabs}
\usepackage{enumerate}

\newcounter{CaseNum}
\newcounter{SubCaseNum}
\newcommand{\Case}{}
\newif\ifContinueCase
\ContinueCasefalse
\newenvironment{CaseDistinction}
{\ifContinueCase
	\else
	\setcounter{CaseNum}{0}
\fi
\renewcommand{\Case}{

	\emph{Case \stepcounter{CaseNum}\arabic{CaseNum}:} }
\ContinueCasefalse
}
{}
\newenvironment{Subcase}
{
\renewcommand{\Case}{

	\emph{Case \stepcounter{SubCaseNum}\arabic{CaseNum}.\arabic{SubCaseNum}:} }
}
{}
\newcounter{InlineEnum}
\newenvironment{inlineenum}
{
\renewcommand{\item}{\refstepcounter{InlineEnum}\emph{(\theInlineEnum)} }
\setcounter{InlineEnum}{0}}
{
}

\usepackage{environ}
	
\RequirePackage{ifthen}
\RequirePackage{xcolor}
\newif\ifShowClosedIssues
\newif\ifShowIssues
\ShowClosedIssuestrue
\ShowIssuestrue
\relax

\colorlet{closedIssue}{green}
\colorlet{openIssue}{red}

\NewEnviron{issue}[2][open]
{\newif\ifShowIssue

\ifShowIssues
	\ifthenelse{\equal{#1}{closed}}
		{\color{closedIssue}\ifShowClosedIssues\ShowIssuetrue\else\ShowIssuefalse\fi}
		{\color{openIssue}\ShowIssuetrue}
\else
	\ShowIssuefalse
\fi
\ifShowIssue

		\noindent
		\rule{\textwidth}{0.05em}\\
		Issue #2: (#1)\\
			\BODY

		\noindent\rule{\textwidth}{0.05em}
	
\fi
}

\usepackage[square,numbers]{natbib}
\usepackage{hyperref}
\hypersetup{colorlinks=true, linkcolor=darkblue, citecolor=darkblue, breaklinks=true}
\usepackage{nameref,thmtools}
\usepackage[sort&compress,nameinlink,noabbrev,capitalize]{cleveref}
	\crefname{rrule}{RR}{Reduction Rules}
\crefname{observation}{Observation}{Observations}

\newcommand{\NameRef}[2]{{\texorpdfstring{\hyperref[#1]{#2} (\cref{#1})}{#2 ({\cref{#1}})}}}
\def\SetLabelName/{Set Label}
\def\SetLabel/{\NameRef{rrule:set label}{\SetLabelName/}}
\def\LowerBoundName/{Lower Bound}
\def\LowerBound/{\NameRef{rrule:lower bound}{\LowerBoundName/}}
\def\DissolveName/{Dissolve Vertex}
\def\Dissolve/{\NameRef{rrule:dissolve vertex}{\DissolveName/}}
\def\ShiftNeighborsName/{Shift Neighbors}
\def\ShiftNeighbors/{\NameRef{rrule:shift neighbours}{\ShiftNeighborsName/}}
\def\LabeledNeighborName/{Labeled Neighbor}
\def\LabeledNeighbor/{\NameRef{rrule:local funnel labeled neighbor}{\LabeledNeighborName/}}
\def\RemoveArcsName/{Remove Arcs}
\def\RemoveArcs/{\NameRef{rrule:remove labeled arcs}{\RemoveArcsName/}}
\def\RemoveSinksName/{Sources and Sinks}
\def\RemoveSinks/{\NameRef{rrule:remove labeled sources and sinks}{\RemoveSinksName/}}
\def\RemoveCyclesName/{Remove Cycles}
\def\RemoveCycles/{\NameRef{rrule:cycle component}{\RemoveCyclesName/}}
\def\BreakCycleName/{Break Cycle}
\def\BreakCycle/{\NameRef{rrule:break cycle}{\BreakCycleName/}}

\newcommand{\FunctionName}[1]{\text{\upshape{\textsf{#1}}}}
\newcommand{\ProblemName}[1]{\textsc{#1}}

\newcommand{\Bo}{\mathcal{O}}

\newcommand{\In}[2][]{\ensuremath{\FunctionName{indeg}_{#1}(#2)}}    
\newcommand{\Out}[2][]{\ensuremath{\FunctionName{outdeg}_{#1}(#2)}}

\newcommand{\InN}[2][]{\ensuremath{\FunctionName{in}_{#1}(#2)}}      

\newcommand{\OutN}[2][]{\ensuremath{\FunctionName{out}_{#1}(#2)}}

\newcommand{\InducedSubgraph}[2]{\ensuremath{#1[#2]}}

\newcommand{\Subgraph}{\subseteq}

\newcommand{\Naturals}{\mathds{N}}
\newcommand{\Size}[1]{\left|{#1}\right|}

\newcommand{\Domain}{\operatorname{Dom}}

\newcommand*{\Abs}[1]{\left| #1 \right|}
\newcommand{\Set}{\coloneqq}

\newcommand{\ProblemDef}[3]{
\noindent{#1}\\
\begin{tabular}{rp{0.8\textwidth}}
\textbf{Input} & #2\\
\textbf{Question} & #3
\end{tabular}
}

\def\FADSlong/{\ProblemName{Funnel Arc Deletion Set}}
\def\FADLlong/{\ProblemName{Funnel Arc Deletion Labeling}}
\def\DFVSlong/{\ProblemName{Directed Feedback Vertex Set}}
\def\DFASlong/{\ProblemName{Directed Feedback Arc Set}}
\def\CVDSlong#1{#1-\ProblemName{Vertex Deletion Set}}
\def\CADSlong#1{#1-\ProblemName{Arc Deletion Set}}
\def\FADS/{\ProblemName{FADS}}
\def\FADL/{\ProblemName{FADL}}
\def\DFVS/{\ProblemName{DFVS}}
\def\DFAS/{\ProblemName{DFAS}}
\def\CVDS#1{#1-\ProblemName{VDS}}
\def\CADS#1{#1-\ProblemName{ADS}}

\newcommand{\Label}[1][]{\ell_{#1}}
\def\ForkL/{\textsc{F}}
\def\MergeL/{\textsc{M}}

\newcommand{\InputLabel}{\ell}
\newcommand{\ReducedLabel}{\ell_r}
\newcommand{\InputSolutionLabel}{\hat{\ell}}
\newcommand{\InputSolutionSet}{S}
\newcommand{\ReducedSolutionLabel}{\hat{\ell}_r}
\newcommand{\ReducedSolutionSet}{S_r}

\RequirePackage{ifthen}
\newif\ifappendix
\appendixtrue

\newcommand{\appsymb}{\ensuremath{\star}}
\newcommand{\appref}[2][\appsymb]{
	\ifappendix
		{\hyperref[app:#2]{#1}}
	\else
	\fi
}

\appendixfalse
\usepackage[toc,page]{appendix}
\usepackage{etoolbox}

\title{A Polynomial Kernel for Funnel Arc Deletion Set}
\author{Marcelo Garlet Milani\thanks{We thank the referees for their numerous helpful comments.}}
\begin{document}
\maketitle
\begin{abstract}
	In \DFASlong/ (\DFAS/) we search for a set of at most $k$ arcs which intersect every cycle in the input digraph.
	It is a well-known open problem in parameterized complexity to decide if \DFAS/ admits a kernel of polynomial size.
	We consider \CADSlong{$\mathcal{C}$} (\CADS{$\mathcal{C}$}), a variant of \DFAS/ where we want to remove at most $k$ arcs from the input digraph in order to turn it into a digraph of a class $\mathcal{C}$.
	In this work, we choose $\mathcal{C}$ to be the class of \emph{funnels}.
	\CADS{\ProblemName{Funnel}} is NP-hard even if the input is a DAG, but is fixed-parameter tractable with respect to $k$.
	So far no polynomial kernels for this problem were known.
	Our main result is a kernel for \CADS{\ProblemName{Funnel}} with $\Bo(k^6)$ many vertices and $\Bo(k^7)$ many arcs, computable in $\Bo(nm)$ time, where $n$ is the number of vertices and $m$ the number of arcs in the input digraph.
\end{abstract}

\newpage

\section{Introduction}

In graph editing problems, we are given a (directed or undirected) graph $G$ and a number $k$, and we search for a set of at most $k$ vertices, edges or arcs whose removal or addition produces a graph with a desired property.
There are several variants of these problems, and in this paper we consider the problem of removing arcs from a digraph in order to obtain a digraph in a given class $\mathcal{C}$.
When $\mathcal{C}$ is the class of all directed acyclic graphs (DAGs), the problem is called \DFASlong/ (\DFAS/).
If we remove vertices instead of arcs, the problem is called \DFVSlong/ (\DFVS/).

There are simple reductions between \DFAS/ and \DFVS/.
We can reduce \DFAS/ to \DFVS/ by taking the line digraph of the input.
Removing a vertex from the reduced instance corresponds to removing an arc from the input instance and vice versa.
For a reduction in the other direction, we split each vertex $v$ into two vertices, say, $v_o$ and $v_i$, connect them with an arc $(v_i, v_o)$ and shift all outgoing arcs of $v$ to $v_o$ and all incoming arcs to $v_i$.
In the context of \emph{parameterized complexity}, such reductions are called \emph{parameterized} as the parameter $k$ is preserved.
Hence, parameterized results are often stated for \DFVS/.

In a breakthrough paper it was proven that there is an algorithm for \DFVS/ with running time $4^kk!\cdot n^{\Bo(1)}$ \cite{chen2008fixed}, showing that the problem is \emph{fixed-parameter tractable} (FPT) with respect to $k$.
After obtaining an FPT result, it is natural to ask if the problem also admits a polynomial \emph{kernel}, that is, if there is a polynomial-time algorithm which reduces the input instance to an instance of size at most $\Bo(k^c)$ for some constant $c$.
Such an algorithm is called a \emph{kernelization} algorithm.

The existence of a polynomial kernel for \DFVS/ is a fundamental open question in the field of parameterized complexity.
One approach towards solving this question is to consider different parametrizations or restrictions of the input digraph.
By considering progressively smaller parameters or more general digraph classes, one can hope to eventually close the gap between the restricted cases and the general case of \DFVS/.

On tournaments, \DFVS/ admits a polynomial kernel \cite{abu2010kernelization};
this was extended to generalizations of tournaments as well \cite{bang2016algorithms}.
When parameterized by solution size $k$ and the size $\ell$ of a treewidth $\eta$-modulator, \DFVS/ admits a kernel of size $(k\cdot\ell)^{\Bo(\eta^2)}$ \cite{lokshtanov2019wannabe}.

One can also restrict the output instead, that is, we can consider \CVDSlong{$\mathcal{C}$} (\CVDS{$\mathcal{C}$}) or \CADSlong{$\mathcal{C}$} (\CADS{$\mathcal{C}$}), where, for a fixed digraph class $\mathcal{C}$, we search for a set of at most $k$ vertices (arcs) whose removal turns the input into a digraph in $\mathcal{C}$.
Unlike \DFVS/ and \DFAS/, \CVDS{$\mathcal{C}$} and \CADS{$\mathcal{C}$} can belong to different complexity classes depending on $\mathcal{C}$: While \CADS{\ProblemName{Out-Forest}} can be solved in polynomial time, \CVDS{\ProblemName{Out-Forest}} is NP-hard \cite{mnich2017polynomial}.
Further, note that even if $\mathcal{C'} \subseteq \mathcal{C}$, a polynomial kernel for \CADS{$\mathcal{C}$} does not immediately imply a polynomial kernel for \CADS{$\mathcal{C'}$}, and the implication also does not work in the other direction.
Indeed, while the problem is trivial when $\mathcal{C}$ is the class of all independent sets or the class of all digraphs, it is NP-hard if $\mathcal{C}$ is the class of DAGs, which contains all independent sets and is a subclass of all digraphs.
In a sense, the complexity landscapes of \CADS{$\mathcal{C}$} and \CVDS{$\mathcal{C}$} are much more fine-grained than the landscape of \DFVS/, and may allow for smaller steps towards more general results.

\CADS{\ProblemName{Out-Forest}} and \CADS{\ProblemName{Pumpkin}} can be solved in polynomial time \cite{mnich2017polynomial}, while \CVDS{\ProblemName{Out-Forest}} and \CVDS{\ProblemName{Pumpkin}} are NP-hard and admit polynomial kernels \cite{agrawal2018kernels,mnich2017polynomial} of size $\Bo(k^2)$ and $\Bo(k^3)$, respectively \cite{agrawal2018kernels}.
\CVDS{$\mathcal{F}_\eta$} admits a polynomial kernel for constant $\eta$, where $\mathcal{F}_\eta$ is the class of all digraphs with (undirected) treewidth at most $\eta$ \cite{lokshtanov2019wannabe}.

In this work we consider \CADS{\ProblemName{Funnel}} and provide a polynomial kernel with $\Bo(k^6)$ vertices and $\Bo(k^7)$ arcs.
A digraph is a funnel if it is a DAG and every source to sink path has an arc which is not in any other source to sink path.
\CADS{\ProblemName{Funnel}} is NP-hard even if the input is DAG, but it can be solved in $\Bo(3^k\cdot(n+m))$ time \cite{Millani18}, where $k$ is the solution size.
Out-forests and pumpkins are also funnels, but there are also dense funnels like complete bipartite digraphs (where all arcs go from the first partition to the second but not back).

Our results rely on characterizations for funnels based on forbidden subgraphs and on a ``labeling'' of the vertices \cite{Millani18}.
We believe the techniques used here can be generalized to other digraph classes which are also similarly characterized, and hope they provide further insight about the classes $\mathcal{C}$ for which \CADS{$\mathcal{C}$} admits a polynomial kernel.

\section{Preliminaries}

A (partial) function $f: A \rightarrow B$ is a set of tuples $(a, f(a)) \in A \times B$ where for every $a \in A$ there is at most one $b \in B$ with $(a,b) \in f$ (that is, $f(a) = b$).
We write $\Domain(f)$ for the set of values $a \in A$ for which $f$ is defined.
Hence, $\emptyset$ is the undefined function, and $f' \supseteq f$ if $f'(x) = f(x)$ for every $x \in \Domain(f)$.
All our functions are \emph{partial}, that is, $\Domain(f)$ is not necessarily $A$.

A \emph{parameterized language} $L$ is \emph{fixed-parameter tractable} with respect to the parameter $k$ if there is some algorithm with running time $f(k)\cdot n^{\Bo(1)}$ deciding whether $(x, k) \in L$, where $f$ is some computable function, $n = \Size{x}$ and $k$ is the parameter (refer to \cite{cygan2015parameterized,downey2013fundamentals} for an introduction to parameterized complexity).
We say that $L$ \emph{admits a problem kernel} if there is a polynomial-time algorithm which transforms an instance $(x,k)$ into an instance $(x', k')$ such that $(x,k) \in L$ if and only if $(x', k') \in L$, $k' \leq k$ and $\Size{x'} \leq f(k)$ for some computable function $f$.
If $f$ is a polynomial, we say that $L$ \emph{admits a polynomial kernel} with respect to $k$.

When describing a kernelization algorithm, it is common to define \emph{reduction rules}.
These rules have a \emph{condition} and an \emph{effect}, and we say that a reduction rule is \emph{applicable} if the condition is true.
The effect of the reduction rule produces a new instance $(x', k')$ of the problem, and a rule is said to be \emph{safe} if $(x', k') \in L$ if and only if the original instance is in $L$.
We refer the reader to \cite{kratsch2014recent,lokshtanov2012kernelization} for surveys on kernelization and to \cite{fomin2019kernelization} for a book on the topic.

We only consider directed graphs (digraphs) without loops or parallel arcs (but we allow arcs in opposite directions) in this paper.
Let $D$ be a digraph.
The set of arcs of $D$ is denoted by $A(D)$, and its set of vertices is $V(D)$.
The set of outneighbors (inneighbors) in $D$ of a vertex $v \in V(D)$ is denoted by $\OutN[D]{v}$ ($\InN[D]{v}$); the outdegree (indegree) of $v$ is $\Out[D]{v} = \Size{\OutN[D]{v}}$ ($\In[D]{v} = \Size{\InN[D]{v}}$).
If the digraph $D$ is clear from context, we omit it from the index.
For a set $U \subseteq V(D)$ we write $\OutN{U}$ for the set $\{\OutN{u} \mid u \in U\} \setminus U$ (and analogously for $\InN{U}$).
A vertex $v$ is a \emph{source} if $\In{v} = 0$, and it is a sink if $\Out{v} = 0$.
We write $H \Subgraph D$ if $H$ a \emph{subgraph} of $D$; the subgraph of $D$ \emph{induced} by $U$ is given by $\InducedSubgraph{D}{U}$.
We write $D - X$ for the operation of deleting a set of vertices or arcs $X$ from $D$.
Similarly, we add a set of arcs or vertices to $D$ with $D + X$.

A \emph{directed acyclic graph} (DAG) is a digraph which does not contain any directed cycle.
A digraph $D$ is a funnel if $D$ is a DAG and for every path $P$ from a source to a sink of $D$ of length at least one there is some arc $a \in A(P)$ such that for any different path $Q$ from a (possibly different) source to a sink we have $a \not\in A(Q)$.
\begin{figure}
	\centering
	\caption{$D_1$, a forbidden subgraph for funnels.}
	\label{fig:d1}
	\begin{tikzpicture}[yscale=0.75]
		\node[gnode, label = left:{\bfseries{}\itshape{}{$u_1$}}, line width = 0.9, fill = white, draw = black]
	(v1) at (0.0, 1.15){};
\node[gnode, label = left:{\bfseries{}\itshape{}{$u_2$}}, line width = 0.9, fill = white, draw = black]
	(v2) at (1.15, 1.15){};
\node[gnode, label = left:{\bfseries{}\itshape{}{$v_0$}}, line width = 0.9, fill = white, draw = black]
	(v3) at (1.15, 0.0){};
\node[gnode, label = right:{\bfseries{}\itshape{}{$v_1$}}, line width = 0.9, fill = white, draw = black]
	(v4) at (2.3, 0.0){};
\node[gnode, label = left:{\bfseries{}\itshape{}{$w_1$}}, line width = 0.9, fill = white, draw = black]
	(v5) at (3.45, 1.15){};
\node[gnode, label = left:{\bfseries{}\itshape{}{$w_2$}}, line width = 0.9, fill = white, draw = black]
	(v6) at (2.3, 1.15){};
\path[edgeto, line width = 0.9, draw = black]
	(v2) to (v3);
\path[edgeto, line width = 0.9, draw = black]
	(v1) to (v3);
\path[edgeto, line width = 0.9, draw = black]
	(v3) to (v4);
\path[edgeto, line width = 0.9, draw = black]
	(v4) to (v6);
\path[edgeto, line width = 0.9, draw = black]
	(v4) to (v5);

	\end{tikzpicture}
\end{figure}
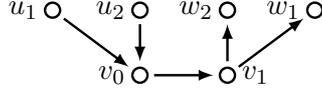
We repeat below several known characterizations for funnels, as they are particularly useful for our results.
\begin{theorem}[\cite{Millani18}, Theorem 1]
	\label{thm:funnel characterization}
	Let $D$ be a DAG.
	The following statements are equivalent.
	\begin{enumerate}[a.]
			\item\label{chr:funnel} $D$ is a funnel.
			\item\label{chr:partition}
				$V(D)$ can be partitioned into two sets $F$ and $M$ such that:
			\begin{inlineenum}
			\item $F$ induces an out-forest;
			\item $M$ induces an in-forest; and
			\item $(M \times F) \cap A(D) = \emptyset$.
			\end{inlineenum}
		\item\label{chr:forbidden subgraph} No digraph in $\mathcal{F} = \{D_i \mid i \in \{0,1, \dots\}\}$ is contained in $D$ as a (not necessarily induced) subgraph, where (see Figure~\ref{fig:d1} for an example)
			\begin{itemize}
				\item $V(D_k) = \{u_1, u_2, w_1, w_2\} \cup \{v_i \mid 0 \leq i \leq k\}$, and 
			\item $A(D_k) = \{(u_1,v_0), (u_2, v_0), (v_k, w_1), (v_k, w_2)\} \cup \{(v_i, v_{i+1}) \mid 1 \leq i \leq k - 1\}$
			\end{itemize}
			\item\label{chr:butterfly} $D$ does not contain $D_0$ as a butterfly minor.				
	\end{enumerate}
\end{theorem}

The digraphs in $\mathcal{F}$ are called \emph{forbidden subgraphs for funnels}.
For a digraph $D$ we define a \emph{labeling} as a function $\InputLabel : V(D) \rightarrow \{\ForkL/, \MergeL/\}$.
We say that $\InputLabel$ is a \emph{funnel labeling} for $D$ if $\Domain(\InputLabel) = V(D)$, the set $F = \{v \in V(D) \mid \InputLabel(v) = \ForkL/\}$ induces an out-forest in $D$, the set $M = \{v \in V(D) \mid \InputLabel(v) = \MergeL/\}$ induces an in-forest in $D$ and $(M \times F) \cap A(D) = \emptyset$.
Due to \cref{thm:funnel characterization}(\ref{chr:partition}), a digraph $D$ is a funnel if and only if there exists a funnel labeling for $D$.

In the \emph{feedback arc set} problem, we are given a digraph $D$ and a $k \in \Naturals$ as an input, and we search for a set $S \subseteq A(D)$ such that $D - S$ is a DAG and $\Size{S} \leq k$.
We consider a variant of this problem where we want $D - S$ to be a funnel instead, which is formally defined below.

\ProblemDef
{\FADSlong/ (\FADS/)}
{A digraph $D$ and a number $k \in \Naturals$.}
{Is there a set $S \subseteq A(D)$ with $\Abs{S} \leq k$ such that $D - S$ is a funnel?}

To make better use of \cref{thm:funnel characterization}(\ref{chr:partition}), we consider a more general problem in which some vertices might already be labeled with \ForkL/ or \MergeL/, and the funnel we obtain in the end must respect this labeling.
Formally, the problem is defined as follows.

\ProblemDef
{\FADLlong/ (\FADL/)}
{A digraph $D$, a labeling $\InputLabel : V(D) \rightarrow \{\ForkL/, \MergeL/\}$ and a number $k \in \Naturals$.}
{Are there a set $S \subseteq A(D)$ and a labeling $\InputSolutionLabel \supseteq \InputLabel$ such that $\InputSolutionLabel$ is a funnel labeling for $D - S$ and $\Size{S} \leq k$?}

We say that $(D,\InputLabel,k)$ is the \emph{input instance} and $(\InputSolutionSet, \InputSolutionLabel)$ is a solution for the input instance.
This more general version of the problem allows us to decide which label a vertex will take and encode this in the instance itself.
While technically not necessary, using \FADL/ instead of \FADS/ simplifies the kernelization algorithm and also the proofs.

\section{Basic reduction rules}
\label{sec:basic}

We construct our kernelization algorithm by defining a series of reduction rules and then showing that, if no reduction rule is applicable, the input size is bounded in a polynomial of $k$.
Our strategy is to partition the vertex set into labeled and unlabeled vertices, then bound the number of unlabeled vertices (\cref{subsec:unlabeled}) and use this to bound the number of labeled vertices (\cref{subsec:labeled}) as well.
In this section we define some reduction rules which are useful both in \cref{subsec:unlabeled} as well as in \cref{subsec:labeled}.
For brevity, we assume that a reduction rule is no longer applicable to the input instance after it has been defined. 

Let $(D, \InputLabel, k)$ be the input instance.
From \cref{thm:funnel characterization}(\ref{chr:forbidden subgraph}) we can see that a funnel has no vertex $v$ with $\In{v} > 1$ and $\Out{v} > 1$.
Further, $\In{v} \leq 1$ if $\InputLabel(v) = \ForkL/$, and $\Out{v} \leq 1$ if $\InputLabel(v) = \MergeL/$.
Hence, by simply counting the number of vertices disrespecting each case, we can obtain a lower bound for the number of arcs that need to be removed from $D$ in order to obtain a funnel.
As removing one arc changes the degree of two vertices, we obtain a bound of at most $2k$ such vertices.
The safety of the following reduction rule follows easily from \cref{thm:funnel characterization}.

\begin{rrule}[\LowerBoundName/]\label{rrule:lower bound}
	Let $V_I \subseteq V(D)$ be the set of vertices with indegree greater than one, let
	$V_O$ be the set of vertices with outdegree greater than one and let $V_X = V_O \cap V_I$.
	Output a trivial ``no'' instance if 
	\begin{align*}
		\sum_{u \in V_O, \Label(u) = \MergeL/}(\Out{u} - 1) + 
		\sum_{u \in V_I, \Label(u) = \ForkL/}(\In{u} - 1)  + \\
		\sum_{u \in V_X, u \not\in \Domain(\Label)}(\min\{\In{u}, \Out{u}\} - 1) > 2k.
	\end{align*}
\end{rrule}

The following reduction rule is based on \cite{Millani18}, with some modifications since the original reduction rule is applied as an intermediate step in an FPT algorithm and is not safe for kernelization.
For certain vertices it is possible to optimally decide which label they should receive in an optimal solution.
For example, vertices with outdegree greater than $k + 1$ can always be labeled with $\ForkL/$, as otherwise we would need to remove at least $k+1$ of its outgoing arcs, which is not possible.
\begin{rrule}[\SetLabelName/] 
	\label{rrule:set label}
	Let~$v \in V(D)$ be an unlabeled vertex.
	
	Set~$\Label(v) \Set \ForkL/$ if at least one of the following is true:
	\begin{inlineenum}
	\item $\In{v} = 0$; 
	\item $v$ has a single inneighbor $u$ and $\Label(u) = \ForkL/$;
	\item there are at least $\In{v}+1$ vertices $u \in \OutN{v}$ with $\Label(u) = \MergeL/$ or $\Label(u) = \ForkL/ \land \In{u} = 1$; or
	\item $\Out{v} > k+1$.
	\end{inlineenum}

	Set~$\Label(v) \Set \MergeL/$ if at least one of the following is true:
	\begin{inlineenum}
	\item $\Out{v} = 0$; 
	\item	$v$ has a single outneighbor $u$ and $\Label(u) = \MergeL/$;
	\item there are at least $\Out{v}+1$ vertices $u \in \InN{v}$ with $\Label(u) = \ForkL/$ or $\Label(u) = \MergeL/ \land \Out{u} = 1$; or
	\item $\In{v} > k + 1$. 
	\end{inlineenum}
\end{rrule}
\begin{proof}[Proof of safety of \SetLabel/]
	Clearly, a solution for the reduced instance is also a solution for the original instance.
	For the other direction, we consider only the case where we set $\Label(v) \Set \ForkL/$, as the other case is symmetric.
	Let $\ReducedLabel$ be the labeling obtained by the reduction rule.
	Let $(\InputSolutionSet, \InputSolutionLabel)$ be a solution for the original instance.
	We set $\ReducedSolutionLabel \Set \InputSolutionLabel$ and $\ReducedSolutionLabel(v) \Set \ForkL/$.
	If $\InputSolutionLabel(v) = \ForkL/$, then clearly $(\InputSolutionSet, \ReducedSolutionLabel)$ is a solution for the reduced instance.
	So assume $\InputSolutionLabel(v) = \MergeL/$.
	This implies that $\Out{v} \leq k+1$, as otherwise $\Abs{\InputSolutionSet} > k$.

	If $\In{v} = 0$, or $\In{v} = 1$ and there is some $u \in \InN{v}$ with $\InputLabel(u) = \ForkL/$,  then $\ReducedSolutionLabel$ is clearly also a funnel labeling for $D - S$.

	Let $U = \{ u \in \OutN{v} \mid \InputLabel(u) = \MergeL/ \text{ or } \InputLabel(u) = \ForkL/ \land \In{u} = 1\}$.
	If $\Abs{U} \geq \In{v} + 1$, we construct an $\ReducedSolutionSet$ from $S$ as follows.
	We add all incoming arcs of $v$ to $\ReducedSolutionSet$ and remove from $\ReducedSolutionSet$ all outgoing arcs $(v,u)$ where $u \in U$.
	Since $\InputSolutionLabel(v) = \MergeL/$, at least $\Out{v} - 1 \geq \In{v}$ many outgoing arcs of $v$ are in $S$.
	Hence, we remove at least $\In{v}$ arcs from $S$ and add at most $\In{v}$.
	Thus, $\Size{\ReducedSolutionSet} \leq \Size{S}$.

	The digraph $D - \ReducedSolutionSet$ does not contain cycles, as all incoming arcs of $v$ were removed, so any cycle in $D - \ReducedSolutionSet$ is also in $D - S$, which is a funnel.
	To see that $\ReducedSolutionLabel$ is a funnel labeling of $D - \ReducedSolutionSet$, first note that we can always keep arcs $(v,u)$ in $D - \ReducedSolutionSet$ where $\Label(u) = \MergeL/$.
	We can also keep arcs $(v,u)$ in $D - \ReducedSolutionSet$ where $\InputLabel(u) = \ForkL/$ and $\In{u} = 1$.
	As $v$ has no incoming arcs in $D - \ReducedSolutionSet$, it lies in an out-forest.
	Hence, $\ReducedSolutionLabel$ is a funnel labeling of $D - \ReducedSolutionSet$.
\end{proof}

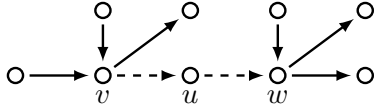
\begin{figure}
	\caption{A digraph which is not a funnel. 
	Removing the arcs $(v,u)$ and $(u,w)$ results in a funnel.}
	\label{fig:example not funnel}
	\centering
	\begin{tikzpicture}[yscale=0.75]
		\node[gnode, line width = 0.9, fill = white, draw = black]
	(v1) at (0.0, 0.0){};
\node[gnode, label = below:{\bfseries{}\itshape{}{$v$}}, line width = 0.9, fill = white, draw = black]
	(v2) at (1.15, 0.0){};
\node[gnode, line width = 0.9, fill = white, draw = black]
	(v3) at (2.3, 1.15){};
\node[gnode, line width = 0.9, fill = white, draw = black]
	(v4) at (1.15, 1.15){};
\node[gnode, label = below:{\bfseries{}\itshape{}{$u$}}, line width = 0.9, fill = white, draw = black]
	(v5) at (2.3, 0.0){};
\node[gnode, label = below:{\bfseries{}\itshape{}{$w$}}, line width = 0.9, fill = white, draw = black]
	(v6) at (3.45, 0.0){};
\node[gnode, line width = 0.9, fill = white, draw = black]
	(v7) at (4.6, 1.15){};
\node[gnode, line width = 0.9, fill = white, draw = black]
	(v8) at (4.6, 0.0){};
\node[gnode, line width = 0.9, fill = white, draw = black]
	(v9) at (3.45, 1.15){};
\path[edgeto, line width = 0.9, draw = black, dashed]
	(v2) to (v5);
\path[edgeto, line width = 0.9, draw = black]
	(v2) to (v3);
\path[edgeto, line width = 0.9, draw = black]
	(v1) to (v2);
\path[edgeto, line width = 0.9, draw = black]
	(v4) to (v2);
\path[edgeto, line width = 0.9, draw = black, dashed]
	(v5) to (v6);
\path[edgeto, line width = 0.9, draw = black]
	(v6) to (v8);
\path[edgeto, line width = 0.9, draw = black]
	(v6) to (v7);
\path[edgeto, line width = 0.9, draw = black]
	(v9) to (v6);

	\end{tikzpicture}
\end{figure}

Replacing an arc in a funnel by a directed path cannot create any cycles nor any forbidden subgraph for funnels.
The next reduction rule reverses this operation: We can contract certain paths where all vertices have in- and outdegree one to a single arc.
However, we cannot replace any such path: In the example in \cref{fig:example not funnel}, if we remove $u$ and add the arc $(v,w)$, then the size of an optimal solution set decreases by one.
Some cases where contracting an arc is safe are identified below.
\begin{rrule}[\DissolveName/]
	\label{rrule:dissolve vertex}
Let $u,v,w$ be a path such that the following is true:
	\begin{inlineenum}
	\item $v,u \in \Domain(\InputLabel)$ implies $\InputLabel(v) = \InputLabel(u)$; and
	\item $v,w \in \Domain(\InputLabel)$ implies $\InputLabel(v) = \InputLabel(w)$.		
	\end{inlineenum}

	If $\In{v} = \Out{v} = 1$ and $(\In{w} = 1 \lor \Out{u} = 1)$, delete the vertex $v$ and add the arc $(u,w)$.
\end{rrule}
\begin{proof}[Proof of safety of \Dissolve/]
	Let $D'$ be the reduced digraph.
	It is easy to see that we can obtain a solution for the reduced instance from a solution for the original instance: If we remove $(u,v)$ or $(v,w)$ from $D$, we can instead remove $(u,w)$ from $D'$.
	As this is equivalent to removing $v$ from $D$, the result is also a funnel and we can keep the same labeling (up to $v$, which is not in $D'$).
	If neither $(u,v)$ nor $(v,w)$ were removed, we simply keep the same arc-deletion set and labeling.

	Now let $(\ReducedSolutionSet, \ReducedSolutionLabel)$ be a solution for the reduced instance.
	We start by setting $\InputSolutionLabel \Set \ReducedSolutionLabel$.
	If $(u,w) \not\in \ReducedSolutionSet$, we set $\InputSolutionLabel(v) \Set \ReducedSolutionLabel(u)$.
	It is easy to see that $\InputSolutionLabel$ is a funnel labeling for $D - \ReducedSolutionSet$.

	If $(u,w) \in \ReducedSolutionSet$, we distinguish two cases.
	If $\Out[D]{u} = 1$, we set $\InputSolutionLabel(v) \Set \InputSolutionLabel(u)$ and $\InputSolutionSet \Set (\ReducedSolutionSet \setminus \{(u,w)\}) \cup \{(v,w)\}$.
	Regardless of whether $\InputSolutionLabel(u) = \MergeL/$ or $\InputSolutionLabel(u) = \ForkL/$, we do not need to remove $(u,v)$.
	Since the neighborhood of $w$ did not change and any cycle in $D - \InputSolutionSet$ is also a cycle in $D' - \ReducedSolutionSet$, we have that $\InputSolutionLabel$ is a funnel labeling for $D - \InputSolutionSet$.

	If $\Out[D]{u} > 1$ and $\In[D]{w} = 1$, we set $\InputSolutionLabel(v) \Set \InputSolutionLabel(w)$ and $\InputSolutionSet \Set (\ReducedSolutionSet \setminus \{(u,w)\}) \cup \{(u,v)\}$.
	As before, we may keep the arc $(w,u)$ in $D - \InputSolutionSet$, and $\InputSolutionLabel$ is a funnel labeling for $D - \InputSolutionSet$.
	Since the case $\Out[D]{u} > 1$ and $\In[D]{w} > 1$ does not occur, this concludes the proof.
\end{proof}

\subsection{Bounding the number of unlabeled vertices}
\label{subsec:unlabeled}

From \LowerBound/ we know there are few vertices with both in- and outdegree greater than one.
In this section we bound the number of unlabeled vertices by considering the remaining unlabeled vertices, that is, vertices $v$ with $\In{v} \leq 1$ or $\Out{v} \leq 1$.
Our strategy is to group such vertices into subgraphs of $D$ with specific properties which we define later, and then develop reduction rules to both bound the maximum number of such subgraphs and also their size in any ``yes'' instance of \FADL/.

Even if the previous reduction rules are not applicable, there can still exist some ``large'' subgraph $H \Subgraph D$ for which there is a ``small'' set $\InputSolutionSet \subseteq A(D)$ such that the weakly-connected component of $H$ is a funnel in $D - \InputSolutionSet$.
Our goal is to bound the size of such subgraphs $H$.

We first define a specific type of subgraph of $D$ which behaves like a funnel in the sense that the degrees of the vertices match \cref{thm:funnel characterization}(\ref{chr:partition}).
We call such subgraphs \emph{local funnels} and formally define them below.
\begin{definition}
	An induced subgraph $H \Subgraph D$ is a \emph{local funnel} in $D$ if $H$ is a funnel, $H$ has only one source and its vertex set can be partitioned into $F \uplus M = V(H)$ such that
		$\In[D]{v} \leq 1$ for all $v \in F$;
		$\Out[D]{v} \leq 1$ for all $v \in M$; and
		$(M \times F) \cap A(H) = \emptyset$.
\end{definition}
Unlike \emph{local} funnels, we might still have to remove many arcs from an \emph{induced} funnel in $D$, as it can have, for example, several vertices $v$ with $\In[D]{v} >1$ and $\Out[D]{v} > 1$.
Our goal is to bound the size of each unlabeled local funnel (that is, each local funnel where none of the vertices have a label) and the number of unlabeled local funnels in $D$.
We start by ``pushing'' as many vertices as we can to the neighborhood of the roots of the in- and out-forests of a local funnel.
Consider for example a path $u,v,w$ as in \cref{fig:shift neighbors}, whose vertices have indegree one but can have higher outdegree.
Intuitively, a cycle containing $v$ and $x$ must also contain $u$.
To destroy this cycle, we can remove the unique incoming arc of $u$, as this will potentially destroy further cycles that contain $u$ but not $v$.
Hence, replacing the arc $(v,x)$ with $(u,x)$ in this case does not change the size of the solution.

By moving vertices in an out-tree towards its root $s$, we increase the outdegree of $s$.
If the outdegree of $s$ increases beyond $k+1$, we can apply \SetLabel/ to $s$, giving it a label.
By further applying \SetLabel/ to the neighbors of $s$ which are in its out-tree, we can label the entire tree.
As we are only considering unlabeled local funnels in this section, we can use the idea above to limit the branching of any in- or out-tree of an unlabeled local funnel.

We provide here a somewhat more general reduction rule which can also be applied if some vertices are labeled.
Later, this reduction rule will again be useful to bound the number of labeled vertices.
However, we need to carefully consider the possible labels of the vertices, as in some cases the rule would not be safe.

\begin{figure}
	\caption{Example application of \ShiftNeighbors/.}
	\label{fig:shift neighbors}
	\centering
	\begin{tikzpicture}[yscale=0.75]
		\node[gnode, label = above:{\bfseries{}\itshape{}{$u$}}, line width = 0.9, fill = white, draw = black]
	(v1) at (1.15, 1.15){};
\node[gnode, label = above:{\bfseries{}\itshape{}{$v$}}, line width = 0.9, fill = white, draw = black]
	(v2) at (2.3, 1.15){};
\node[gnode, label = above:{\bfseries{}\itshape{}{$w$}}, line width = 0.9, fill = white, draw = black]
	(v3) at (3.45, 1.15){};
\node[gnode, label = right:{\bfseries{}\itshape{}{$x$}}, line width = 0.9, fill = white, draw = black]
	(v7) at (2.3, 0.0){};
\node[gnode, line width = 0.9, fill = white, draw = black]
	(v12) at (0.0, 1.15){};
\path[edgeto, line width = 0.9, draw = black]
	(v1) to (v2);
\path[edgeto, line width = 0.9, draw = black]
	(v2) to (v3);
\path[edgeto, line width = 0.9, draw = black]
	(v2) to (v7);
\path[edgeto, line width = 0.9, draw = black]
	(v12) to (v1);

	\end{tikzpicture}
	\quad
	\begin{tikzpicture}
		\node (a) at (0.5,-1) {};
		\node (a) at (0.5,1) {};
		\draw[-latex, line width=5] (0,0) -- (1,0);
	\end{tikzpicture}
	\quad
	\begin{tikzpicture}[yscale=0.75]
		\node[gnode, label = above:{\bfseries{}\itshape{}{$u$}}, line width = 0.9, draw = black]
	(v1) at (1.15, 1.15){};
\node[gnode, label = above:{\bfseries{}\itshape{}{$v$}}, line width = 0.9, draw = black]
	(v2) at (2.3, 1.15){};
\node[gnode, label = above:{\bfseries{}\itshape{}{$w$}}, line width = 0.9, draw = black]
	(v3) at (3.45, 1.15){};
\node[gnode, label = right:{\bfseries{}\itshape{}{$x$}}, line width = 0.9, draw = black]
	(v7) at (2.3, 0.0){};
\node[gnode, line width = 0.9, draw = black]
	(v12) at (0.0, 1.15){};
\path[edgeto, line width = 0.9, draw = black]
	(v1) to (v2);
\path[edgeto, line width = 0.9, draw = black]
	(v2) to (v3);
\path[edgeto, line width = 0.9, draw = black]
	(v1) to (v7);
\path[edgeto, line width = 0.9, draw = black]
	(v12) to (v1);

	\end{tikzpicture}
\end{figure}
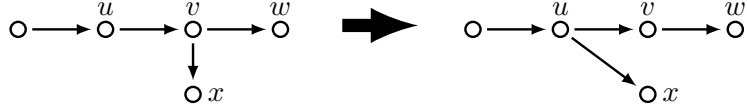

\begin{rrule}[\ShiftNeighborsName/]
\label{rrule:shift neighbours}
	Let $u,v,w$ be a path.
\begin{itemize}
	\item If $\In{u} = \In{v} = \In{w} = 1$, $(u, \MergeL/) \not\in \InputLabel$, $(v, \MergeL/) \not\in \InputLabel$  and there is an $x \in \OutN{v} \setminus \OutN{u}$ with $w \neq x$, then remove the arc $(v,x)$ and add the arc $(u,x)$.
	\item If $\Out{u} = \Out{v} = \Out{w} = 1$, $(v, \ForkL/)\not\in \InputLabel$, $(w, \ForkL/) \not\in \InputLabel$ and there is an $x \in \InN{v} \setminus \InN{w}$ with $u \neq x$, then remove the arc $(x,v)$ and add the arc $(x,w)$.
\end{itemize}
\end{rrule}

Before proving that \ShiftNeighbors/ is safe, we need two simple observations about certain cases where we can safely exchange two arcs or add an arc.

\begin{observation}
	\label{obs:flip incoming arc merge}
	Let $H$ be a funnel with funnel labeling $\Label$ and let $x,u,v \in V(H)$ such that $(v,x) \in A(H)$, $(u,x) \not\in A(H)$ and at least one of the following is true:
	\begin{inlineenum}
		\item $\Label(u) = \ForkL/$; or
		\item $\Label(u) = \MergeL/ = \Label(v)$ and $\Out[H]{u} = 0$.
	\end{inlineenum}
	Let $H' = H - (v,x) + (u,x)$.
	Then $\Label$ is also a funnel labeling for $H'$ if $H'$ is a DAG.
\end{observation}
\begin{proof}
	Assume $H'$ is a DAG.
	\begin{CaseDistinction}
	\Case $\Label(u) = \ForkL/$.
	If $\InputLabel(x) = \MergeL/$, then both $H + (u,x)$ and $H - (v,x) + (u,x)$ are funnels.
	If $\InputLabel(x) = \ForkL/$, then $H'$ is the result of switching the unique inneighbour of $x$ in the out-forest induced by vertices labeled with $\ForkL/$.
	This is clearly an out-forest, and hence $\Label$ is a funnel labeling for $H'$.
	\Case $\InputLabel(u) = \MergeL/ = \InputLabel(v)$ and $\Out[H]{u} = 0$.
		As $(v,x) \in A(H)$, we have $\InputLabel(x) = \MergeL/$.
		Since $\Out[H]{u} = 0$, $\InputLabel$ is a funnel labeling for $H + (u,x)$ and, hence, also for $H'$.\qedhere
	\end{CaseDistinction}
\end{proof}

\begin{observation}
	\label{obs:flip arc cycle}
	Let $H$ be a DAG and $x,u,v \in V(H)$ such that $\{u\} = \InN{v}$.
	Then $H + (u,x)$ contains a cycle if and only if $H + (v,x)$ contains a cycle.
\end{observation}
\begin{proof}
	Assume $H + (v,x)$ contains a cycle $C$.
	As $\{u\} = \InN{v}$, we get $(u,v) \in A(C)$.
	Hence, replacing $(u,v)$ and $(v,x)$ with $(u,x)$ in $C$ produces a cycle in $H + (u,x)$.

	If $H + (u,x)$ contains a cycle $C$, then we can replace $(u,x)$ by the path from $u$ to $x$ in $H + (v,x)$ going through $v$.
	This constructs a cycle in $H + (v,x)$, as desired.
\end{proof}

\begin{proof}[{Proof of safety of \ShiftNeighbors/}]
	Consider the case where $\In{u} = \In{v} = \In{w} = 1$, $(u, \MergeL/) \not\in \InputLabel$, $(v, \MergeL/) \not\in \InputLabel$ and there is an $x \in \OutN{v} \setminus \OutN{u}$ with $w \neq x$.
	The other case follows analogously. 
	Let $(D', \InputLabel, k)$ be the reduced instance and $(\ReducedSolutionSet, \ReducedSolutionLabel)$ be a solution for it.
	We construct a solution $(S, \InputSolutionLabel)$ for the input instance $(D, \InputLabel, k)$.

	First observe that, if $(u,x) \in \ReducedSolutionSet$, we can replace it with $(v,x)$ in $S$, which means that $D' - \ReducedSolutionSet$ and $D - S$ are isomorphic.
	By setting $\InputSolutionLabel \Set \ReducedSolutionLabel$, we obtain the desired solution.
	If $(u,x) \not \in \ReducedSolutionSet$, we consider the following cases.
\begin{CaseDistinction}
\Case $\ReducedSolutionLabel(v) = \ForkL/$.
	We set $\InputSolutionLabel \Set \ReducedSolutionLabel$ and $S \Set \ReducedSolutionSet$.
	Let $D^\star = D - S$.
	Clearly, $D^\star = D' - \ReducedSolutionSet - (u,x) + (v,x)$.
	As $u$ is the only inneighbor of $v$, from \cref{obs:flip arc cycle} we know $D^\star$ is a DAG.
	From \cref{obs:flip incoming arc merge}, we know that $\InputSolutionLabel = \ReducedSolutionLabel$ is a funnel labeling for $D^\star$.
\Case $\ReducedSolutionLabel(v) = \MergeL/ = \ReducedSolutionLabel(u)$.
	If $D' - \ReducedSolutionSet + (u,v)$ is a DAG, we can assume that $(u,v) \not\in \ReducedSolutionSet$, implying $(u,x) \in \ReducedSolutionSet$ (which was already considered).
	
	If $D' - \ReducedSolutionSet + (u,v)$ is not a DAG, then it contains a cycle with $v$ and $u$, implying $(u,v) \in \ReducedSolutionSet$.
	In particular, $\In[D' - \ReducedSolutionSet]{v} = 0$.
	We set $\InputSolutionLabel \Set \ReducedSolutionLabel$ and $\InputSolutionLabel(v) \Set \ForkL/$.
	Clearly, $\InputSolutionLabel$ is a funnel labeling for $D' - \ReducedSolutionSet$.
	From \cref{obs:flip incoming arc merge} we have that $\InputSolutionLabel$ is a funnel labeling for $D - \ReducedSolutionSet$ as well.

\Case $\ReducedSolutionLabel(v) = \MergeL/$ and $\ReducedSolutionLabel(u) = \ForkL/$.
	We set $\InputSolutionLabel \Set \ReducedSolutionLabel$, $\InputSolutionLabel(v) \Set \ForkL/$ and $S \Set \ReducedSolutionSet$.
	As $\{u\} = \InN[D]{v}$ and $\ReducedSolutionLabel(u) = \ForkL/$, $\InputSolutionLabel$ is a funnel labeling for $D' - \ReducedSolutionSet$.
	Let $D^\star = D - S$.

	From \cref{obs:flip arc cycle} we know $D^\star = D' - \ReducedSolutionSet - (u,x) + (v,x)$ is a DAG since $D' - \ReducedSolutionSet$ is a DAG.
	Hence, from \cref{obs:flip incoming arc merge} we obtain that $(\ReducedSolutionSet, \InputSolutionLabel)$ is a solution for the input instance.
\end{CaseDistinction}
	In all cases a solution for the reduced instance implies a solution for the original instance.

	Now assume there is a solution $(\InputSolutionSet, \InputSolutionLabel)$ for the original instance.
	We show that there is solution $(\ReducedSolutionSet, \ReducedSolutionLabel)$ for the reduced instance.
	As in the previous direction, if $(v,x) \in S$, we can replace it with $(u,x)$ and obtain the desired solution.
	So assume $(v,x) \not \in S$.

	If $(u,v) \in S$, let $S_1 = S \cup \{(y,u)\}$, where $\{y\} = \InN{u}$.
	Clearly, $\InputSolutionLabel$ is a funnel labeling for $D - S_1$.
	We set $\ReducedSolutionLabel \Set \InputSolutionLabel$ and $\ReducedSolutionLabel(u) \Set \ForkL/$.
	As $\In[D - S_1]{u} = 0$, $\ReducedSolutionLabel$ is also a funnel labeling for $D - S_1$.
	From \cref{obs:flip incoming arc merge} we know that $\ReducedSolutionLabel$ is a funnel labeling for $D_1 = D' - S_1$.
	Since $\In[D_1]{v} = 0 = \In[D_1]{u}$ and $\ReducedSolutionLabel(u) = \ForkL/$, we have that $\ReducedSolutionLabel$ is a funnel labeling for $D_1 + (u,v)$.
	Hence, $(\InputSolutionSet \setminus \{(u,v)\}, \ReducedSolutionLabel)$ is a solution for the reduced instance.
	
	In the following we consider the remaining cases where $\{(u,v), (v,x)\} \cap S =~\emptyset$.
	Note that the case $\InputSolutionLabel(u) = \MergeL/$ and $\InputSolutionLabel(v) = \ForkL/$ does not happen under this assumption.
\begin{CaseDistinction}
\Case $\InputSolutionLabel(v) = \ForkL/ = \InputSolutionLabel(u)$.
	We set $\ReducedSolutionLabel \Set \InputSolutionLabel$ and $\ReducedSolutionSet \Set S$.
	Clearly, $D' - \ReducedSolutionSet = D - \InputSolutionSet - (v,x) + (u,x)$.
	From \cref{obs:flip arc cycle} there is no cycle in $D - \InputSolutionSet + (u,x)$ and, hence, $D' - \ReducedSolutionSet$ is a DAG.
	Thus, from \cref{obs:flip incoming arc merge} we have that $\ReducedSolutionLabel$ is a funnel labeling for $D' - \ReducedSolutionSet$.
\Case $\InputSolutionLabel(v) = \MergeL/ = \InputSolutionLabel(u)$.
	Since $(v,x) \not\in \InputSolutionSet$, we have $(v,w) \in \InputSolutionSet$ and $\InputSolutionLabel(x) = \MergeL/$.
	Further, we know that $D' - \InputSolutionSet$ is a DAG due to \cref{obs:flip arc cycle}.
	Let $S_1 = \InputSolutionSet \cup \{(u,v)\}$.
	Clearly, $\InputSolutionLabel$ is a funnel labeling for $D - S_1$, and $D' - S_1$ is also a DAG.
	From \cref{obs:flip incoming arc merge} we have that $\InputSolutionLabel$ is a funnel labeling for $D' - S_1$.

	We set $\ReducedSolutionLabel \Set \InputSolutionLabel$ and $\ReducedSolutionLabel(v) \Set \ForkL/$.
	Since $\In[D' - S_1]{w} = 0 = \In[D' - S_1]{v}$, we have that $\ReducedSolutionLabel$ is a funnel labeling for $D' - S_1 + (v,w)$, regardless of the label of $w$.
	By setting $\ReducedSolutionSet \Set (S \setminus \{(v,w)\}) \cup \{(u,v)\}$, we get that $\ReducedSolutionLabel$ is a funnel labeling for $D' - \ReducedSolutionSet$ and $\Size{\ReducedSolutionSet} \leq \Size{S}$.
\Case $\InputSolutionLabel(v) = \MergeL/$ and $\InputSolutionLabel(u) = \ForkL/$.
	Let $\ReducedSolutionSet = S$ and $\ReducedSolutionLabel = \InputSolutionLabel$.
	Since $(u,v) \not\in \ReducedSolutionSet$, from \cref{obs:flip arc cycle} we know that $D - \ReducedSolutionSet - (v,x) + (u,x)$ is a DAG.
	From \cref{obs:flip incoming arc merge} we have that $\ReducedSolutionLabel$ is a funnel labeling for $D' - \ReducedSolutionSet$.
\end{CaseDistinction}

In all cases we found a solution $(\ReducedSolutionSet, \ReducedSolutionLabel)$ for the reduced instance, concluding the proof.
\end{proof}

It is not always possible to exhaustively apply \ShiftNeighbors/: If $u,v,w$ forms a cycle, we would shift $x$ indefinitely through this cycle.
To prevent this from happening, we need the following reduction rule:

\begin{rrule}[\BreakCycleName/]\label{rrule:break cycle}
	Let $C$ be a cycle in $D$.
	If every vertex in $C$ has indegree (outdegree) one and either every vertex in $C$ is unlabeled or every vertex in $C$ is labeled with $\ForkL/$ ($\MergeL/$), then delete one arc of $C$ and decrease $k$ by one.
\end{rrule}
\begin{proof}[Proof of safety of \BreakCycle/]
	Let $(v,u)$ be the arc removed by the reduction rule.
	Clearly, a solution for the reduced instance together with the arc $(v,u)$ is a solution for the original instance.
	Let $(\InputSolutionSet, \InputSolutionLabel)$ be a solution for the original instance, and assume that $(v,u) \not\in \InputSolutionSet$.
	Let $(w,x)$ be an arc of $C$ contained in $\InputSolutionSet$.
	Without loss of generality, we assume that $(w,x)$ is the only incoming arc of $x$.
	The case where it is the only outgoing arc of $w$ follows analogously.

	We can assume that $\InputSolutionLabel(v) = \ForkL/$ for all $v \in V(C)$: If they were not labeled by $\InputLabel$ when the rule was applied, then by repeatedly applying \SetLabel/ (starting with $x$) we can label them with $\ForkL/$.
	Because $\In[D]{v} = 1$ for every $v \in C$, it follows that $C$ is the only cycle in $D$ using the arc $(w,x)$.
	Hence, $D' = D - \InputSolutionSet + (w,x) - (v,u)$ is a DAG.
	Further, as $\InputSolutionLabel(w) = \InputSolutionLabel(x) = \ForkL/$, it is easy to see that $\InputSolutionLabel$ is a funnel labeling for $D'$.
\end{proof}

If \ShiftNeighbors/ is not applicable, then many vertices in a long path $P$ in a local funnel must share a common out- or inneighbor $w$.
However, from \SetLabel/ we know that $w$ receives a label if it has too many neighbors.
The next and final reduction rule needed for bounding the number of unlabeled vertices exploits this property and allows us to label some vertex $u$ in $P$ if its predecessor $v$ in $P$ is adjacent to a labeled vertex $w$.

\begin{rrule}[\LabeledNeighborName/]
	\label{rrule:local funnel labeled neighbor}
	Let $(v,u)$ be an arc between unlabeled vertices.
	Set $\Label(u) \Set \ForkL/$ if $\In{u} = \In{v} = 1$ and $\exists w \in \OutN{v} : \Label(w) = \MergeL/$.
	Set $\Label(v) \Set \MergeL/$ if $\Out{u} = \Out{v} = 1$ and $\exists w \in \InN{u} : \Label(w) = \ForkL/$.
\end{rrule}
\begin{proof}[Proof of safety of \LabeledNeighbor/]
	Assume, without loss of generality, that the first case of the rule was applied.
	The proof for the second case follows analogously (note that it is not possible for both cases to be applied simultaneously).
	Let $(D, \ReducedLabel, k)$ be the reduced instance. 
	First note that $\ReducedLabel \supseteq \InputLabel$, which means that a solution for the reduced instance is already a solution for the original instance.
	Hence, it suffices to show that a solution $(\InputSolutionSet,\InputSolutionLabel)$ for the original instance implies a solution $(\ReducedSolutionSet, \ReducedSolutionLabel)$ for the reduced instance.
	
	If $\InputSolutionLabel(u) = \ForkL/$, we set $\ReducedSolutionLabel \Set \InputSolutionLabel$ and $\ReducedSolutionSet \Set \InputSolutionSet$ and we are done.
	So assume that $\InputSolutionLabel(u) = \MergeL/$.

	\begin{CaseDistinction}
	\Case $(v,u) \in S$.
	We set $\ReducedSolutionSet \Set \InputSolutionSet$, $\ReducedSolutionLabel \Set \InputSolutionLabel$ and $\ReducedSolutionLabel(u) \Set \ForkL/$.
	As $\In[D - \InputSolutionSet]{u} = 0$, we know that $\ReducedSolutionLabel$ is also a funnel labeling for $D - \InputSolutionSet$.
	\Case $(v,u) \not\in S$ and $\InputSolutionLabel(v) = \ForkL/$.
		We set $\ReducedSolutionSet \Set \InputSolutionSet$, $\ReducedSolutionLabel \Set \InputSolutionLabel$ and $\ReducedSolutionLabel(u) \Set \ForkL/$.
		As $\ReducedSolutionLabel(v) = \ForkL/ = \ReducedSolutionLabel(u)$, we may keep the arc $(v,u)$ and $\ReducedSolutionLabel$ is a funnel labeling for $D - \ReducedSolutionSet$.
	\Case $(v,u) \not\in S$ and $\InputSolutionLabel(v) = \MergeL/$.
		Then $(v,w) \in \InputSolutionSet$.
		We set $\ReducedSolutionLabel \Set \InputSolutionLabel$, $\ReducedSolutionLabel(u) \Set \ForkL/$, $\ReducedSolutionLabel(v) \Set \ForkL/$, $\ReducedSolutionSet \Set (\InputSolutionSet \setminus \{(v,w)\}) \cup \{(y,v)\}$, where $y$ is the unique inneighbor of $v$.

		The digraph $D - \ReducedSolutionSet$ is a DAG: if it has a cycle, the cycle would have to use the arc $(v,w)$, yet $\In[D - \ReducedSolutionSet]{v} = 0$, a contradiction.
		We now argue that $\ReducedSolutionLabel$ is a funnel labeling for $D - \ReducedSolutionSet$.
		Since $\In[D - \ReducedSolutionSet]{v} = 0$, $\In[D - \ReducedSolutionSet]{u} = 1$ and $\ReducedSolutionLabel(u) = \ForkL/$, the vertex $v$ is the unique inneighbor of $u$ in the out-forest of the funnel $D - \ReducedSolutionSet$.
		Finally, as $\ReducedSolutionLabel(w) = \MergeL/$, the arc $(v,w)$ is allowed in the funnel.
		Hence, $\ReducedSolutionLabel$ is a funnel labeling for $D - \ReducedSolutionSet$.
	\end{CaseDistinction}
	In all cases we find a solution $(\ReducedSolutionLabel$, $\ReducedSolutionSet)$ for the reduced instance, concluding the proof.
\end{proof}

\begin{lemma}
	\label{lemma:local funnel paths}
	Let $s$ be some source (sink) of some unlabeled local funnel $H$ in the reduced digraph $D$.
	Let $P_1, P_2, \dots P_{a}$ be a sequence of paths in $H$ starting (ending) at $s$ such that $\In{u} \leq 1$ $(\Out{u} \leq 1)$ for each $u$ in each $P_i$, and $V(P_j) \not\subseteq V(P_i)$ for all $1 \leq i,j \leq a$ where $i \neq j$.
	Let $E$ be the set of end (start) points of all $P_i$.
	Then all of the following hold.
	\begin{enumerate}[(1)]
		\item \label{item: no subdivided arc}
			$\Out{u} > 1$ $(\In{u} > 1)$ for any inner vertex $u$ of any $P_i$.
		\item \label{item:s reaches outneighbours}
			$\OutN{\bigcup_{i=1}^aV(P_i) \setminus E} \subseteq \OutN{s}$ $(\InN{\bigcup_{i=1}^aV(P_i) \setminus E} \subseteq \InN{s})$,
		\item \label{item: no branchings}
			$V(P_i) \cap V(P_j) = \{s\}$ for each $1 \leq i,j \leq a$ where $i \neq j$, and
		\item \label{item:few paths}
			$a \leq k+1$ and $\Abs{V(P_i)} \leq k + 2$ for each $1 \leq i \leq a$.
		\end{enumerate}
\end{lemma}
\begin{proof}
	We consider the case where $s$ is a source of $H$.
	The other case follows analogously.

	Let $u$ be some inner vertex of some $P_i$ and $w$ the unique outneighbor of $u$ in $P_i$.
	By assumption on $P_i$, we have $\In[D]{w} = 1$.
	As \Dissolve/ is not applicable, we have that $\Out{u} > 1$ (proving (\ref{item: no subdivided arc})).
	In particular, $u$ has some outneighbor $x$ not in $P_i$.

	Let $v$ be the inneighbor of $u$ in $P_i$.
	Since $\In[D]{v} = \In[D]{u} = \In[D]{w} = 1$ and \ShiftNeighbors/ is not applicable, we have $x \in \OutN[D]{v}$.
	By repeating this argument to the predecessors of $u$ in $P_i$, we prove (\ref{item:s reaches outneighbours}) (and also that $a \leq k+1$, as $\Out[D]{s} \leq k+1$ due to \SetLabel/).

	Assume there are two paths $P_i$ and $P_j$ intersecting at more than one vertex.
	Let $u$ be the last vertex of the intersection.
	Note that, if $u$ is the last vertex of $P_i$ or $P_j$, then one path has to contain the other.
	Hence, $u$ has two outneighbors $w_i$ and $w_j$ lying on $P_i$ and $P_j$, respectively, and $w_i \neq w_j$.
	But due to (\ref{item:s reaches outneighbours}), we have $w_i, w_j \in \OutN[D]{s}$, implying $\In[D]{w_i} > 1$ and $\In[D]{w_j} > 1$, a contradiction to our assumptions on $P_i$ and $P_j$ (proving (\ref{item: no branchings})).

	Let $v_1, v_2, \dots, v_m$ be the	sequence of vertices of a path $P_i$.
	From (\ref{item: no subdivided arc}) we know that there is some $w \in \OutN[D]{v_{m-1}}$ outside of $P_i$.
	We also have $w \in \OutN[D]{v_j}$ for all $1 \leq j \leq m-1$, implying $\In[D]{w} \geq m-1$.
	If $m-1 > k+1$, then $\InputLabel(w) = \MergeL/$, as \SetLabel/ is not applicable.
	However, as $\In[D]{v_{m-1}} = 1 = \In[D]{v_{m-2}}$, $w \in \OutN[D]{v_{m-2}}$ and \LabeledNeighbor/ is not applicable, we have $\InputLabel(v_{m-1}) = \ForkL/$, a contradiction to the assumption that $H$ is unlabeled.
	Hence, $m - 1 \leq k+1$, implying $\Size{V(P_i)} \leq k + 2$ (proving (\ref{item:few paths})).
\end{proof}

\begin{lemma}
	\label{lemma:small local funnel}
	Let $H$ be an unlabeled local funnel in $D$.
	Then $\Size{V(H)} \in \Bo(k^3)$.
\end{lemma}
\begin{proof}
	Let $s$ be the source of $H$.
	Consider a partitioning of the vertices of $H$ into an out-tree (since $H$ has only one source) and an in-forest where the out-tree is maximal.
	Let $P_1, P_2, \dots, P_a$ be a sequence of paths such that the out-tree is the union of all $P_i$.
	From \cref{lemma:local funnel paths} we know that $a \leq k + 1$ and that $V(P_i) \cap V(P_j) = \{s\}$ for all $i \neq j$.
	Let $v_i$ be the endpoint of $P_i$ which is not $s$ and let $X = \bigcup_{i=1}^a\OutN{v_i}$.
	As \SetLabel/ is not applicable, we have $\Size{X} \leq a\cdot k \leq k^2 + k$.
	Further, as $\Size{V(P_i)} \leq k+1$, the out-tree of $H$ has at most $(k+1)^2$ many vertices.

	Let $Y$ be the set of sinks of $H$ lying on its in-forest.
	Since $H$ has only one source $s$, for every sink $t \in Y$ there is a path $Q$ from $s$ to $t$.
	Let $Q_1, Q_2, \dots, Q_b$ be the set of all paths from $s$ to each sink in $Y$, and let $R_i$ be the subpath of $Q_i$ contained in the in-forest of $H$.

	Let $Q_i$ be one of such paths, and let $u$ be the first vertex of $R_i$ (which is not in any $P_j$).
	Note that this implies $\In{u} > 1$, otherwise the out-tree would not be maximal.
	Due to \cref{lemma:local funnel paths} we have that no other $R_\ell$ contains $u$ and if $u \not\in x$, then $u \in \OutN{s}$.
	Since \SetLabel/ is not applicable and each distinct $R_\ell$ implies the existence of a distinct outneighbor of $s$, there are at most $k$ paths $R_\ell$ not ending in a vertex in $X$.
	By definition, all inneighbors of any vertex of $X$ lie in some $P_i$.
	This implies that there are at most $k$ paths $Q_\ell$ not containing any vertex of $X$.

	If $Q_i$ contains a vertex of $X$, then $R_i$ ends on a vertex $u \in X$.
	Due to \cref{lemma:local funnel paths}, no other $R_\ell$ contains $u$.
	As $\Size{X} \leq k^2 + k$, we have that there are at most $k^2 + k$ paths $R_\ell$ which contain some vertex of $X$.
	Adding both cases, we obtain that there are at most $k^2 + 2k$ paths $Q_\ell$.

	Due to \cref{lemma:local funnel paths}, the subpath $R_i$ of $Q_i$ has at most $k+1$ vertices.
	Since the in-forest of $H$ is the union over all $R_i$, we have that this in-forest has at most $(k+1)(k^2 + k) = k^3 + 2k^2 + k$ many vertices.
	Thus, $\Size{V(H)} \leq (k+1)^2 + (k+1)(k^2 + 2k) \in \Bo(k^3)$, concluding the proof.
\end{proof}

We conclude by bounding the number of maximal vertex-disjoint unlabeled local funnels in $D$.
Since we can always partition unlabeled vertices with in- or outdegree at most one into local funnels, by bounding the number of local funnels in such a partitioning, together with the bound on the size of each local funnel, we obtain a bound for the number of unlabeled vertices with in- or outdegree at most one.

Let $\mathcal{H} = \{H_1, H_2, \dots H_{a}\}$ be a set of maximal vertex-disjoint unlabeled local funnels in $D$
(in this context, maximal means that $H_i \cup H_j$ is not a local funnel for any two distinct $H_i,H_j \in \mathcal{H}$).
Let $s_i$ be the unique source of $H_i$ for each $i$.
We now show that, if there is a solution removing at most $k$ arcs, then $\Size{\mathcal{H}}$ is ``small''.
By contraposition this means that, if $\Size{\mathcal{H}}$ is ``large'', then we have a ``no'' instance and can stop the kernelization process.

We start with the simple observation that cycles intersecting inside a local funnel must also intersect outside it.
\begin{observation}
	\label{obs:vertex disjoint cycles local funnel}
	Let $C_i$ and $C_j$ be two distinct cycles in $D$ such that $V(C_i) \cap V(C_j) \subseteq H_\ell$ for some $H_\ell \in \mathcal{H}$.
	Then $V(C_i) \cap V(C_j) = \emptyset$.
\end{observation}
\begin{proof}
	Assume towards a contradiction that $V(C_i) \cap V(C_j) \neq \emptyset$.
	Let $v \in V(C_i) \cap V(C_j)$ such that the predecessor of $v$ in $C_i$ is different from the predecessor of $v$ in $C_j$.
	Then $\In{v} > 1$.
	As $H_\ell$ is a local funnel, $v$ can only reach one sink $t$ of $H_\ell$, implying that $t$ is in both $C_i$ and $C_j$.
	The unique out-neighbor of $t$ is however not in $H_\ell$, but it has to be in both $C_i$ and $C_j$, a contradiction to the assumption that $V(C_i) \cap V(C_j) \subseteq V(H_\ell)$.
\end{proof}

We partition the set of maximal unlabeled local funnels $\mathcal{H}$ into three sets
\begin{inlineenum}
\item $\mathcal{F} = \{H_i \in \mathcal{H} \mid \text{there is some } v \in V(H_i) \text{ with }\Out[D]{v} > 1\}$;
\item	$\mathcal{M} = \{H_i \in \mathcal{H} \mid \In[D]{s_i} > 1\}$; and
\item	$\mathcal{X} = \{H_i \in \mathcal{H} \mid \In[D]{s_i} = 1 \text{ and } \forall v \in V(H_i) : \Out[D]{v} = 1\}$.
\end{inlineenum}

\begin{lemma}
	\label{lemma:bound |X|}
	If there is a solution $(\InputSolutionSet, \InputSolutionLabel)$ for $(D, \InputLabel, k)$, then
	$\Size{\mathcal{X}} \leq 2k^2$.
\end{lemma}
\begin{proof}
	Let $H_i \in \mathcal{X}$ and $u$ be the unique inneighbor of $s_i$.
	Note that $\Out[D]{s_i} = 1$.
	As \Dissolve/ is not applicable, we have that $\Out[D]{u} > 1$ and $\In[D]{w} > 1$, where $w$ is the unique outneighbor of $s_i$.
	\begin{CaseDistinction}
	\Case $u \in \Domain(\InputLabel)$.
		Then $\InputLabel(u) = \MergeL/$ since \SetLabel/ is not applicable.
		As $\Out{u} > 1$, each $H_j \in \mathcal{X}$ with $s_j \in \OutN[D]{u}$ requires one more arc of $u$ to be in $\InputSolutionSet$.
	\Case $u \not \in \Domain(\InputLabel)$.
		If $\In[D]{u} = 1$, then there is some $v_i \in V(H_i)$ such that $(v_i,u) \in A(D)$, otherwise $H_i$ would not be maximal.
		Hence, there is a cycle $C_i$ containing $u,s_i$ and $v_i$.
		If there is any other $H_j \in \mathcal{X}$ with $s_j \in \OutN[D]{u}$ and with some $v_j \in V(H_j)$ such that $(v_j, u)$, then the cycle $C_j$ containing $u,s_j$ and $v_j$ is arc-disjoint to the cycle $C_i$ due to \cref{obs:vertex disjoint cycles local funnel}.
		Thus, $\InputSolutionSet$ must contain at least one arc of each such $C_j$, implying there are at most $k$ local funnels $H_j$ that fall into this case.

		If $\In[D]{u} > 1$, one arc of $u$ is in $\InputSolutionSet$ as $\Out[D]{u} > 1$.
		Further, $\Out[D]{u} \leq k$.
		This means that there are at most $k$ local funnels $H_j \in \mathcal{X}$ with $s_j \in \OutN[D]{u}$.
		As there can be at most $2k$ such vertices $u$, we have that there are at most $2k^2$ local funnels $H_j \in \mathcal{X}$ which fall into this case.
	\end{CaseDistinction}
	In the worst case, we have $\Size{\mathcal{X}} \leq \max\{k+1, 2k^2\} \leq 2k^2$.
\end{proof}

\begin{lemma}
	\label{lemma:bound |F|}
	If there is a solution $(\InputSolutionSet, \InputSolutionLabel)$ for $(D, \InputLabel, k)$, then $\Size{\mathcal{F}} \leq 2k^2 + 3k$.
\end{lemma}
\begin{proof}
	Let $H_i \in \mathcal{F}$ and let $u$ be the unique inneighbor of $s_i$ in $D$.
	Assume $u$ is in some local funnel $H_j \in \mathcal{H}$ and let $D_i = \InducedSubgraph{D}{V(H_i) \cup V(H_j)}$.
	\begin{CaseDistinction}
		\Case	$D_i$ is a DAG.
		Then there are $w_j \in V(H_j)$ and $w_i \in V(H_i)$ such that $\In[D]{w_j} > 1$, $\Out[D]{w_i} > 1$ and there is a path $P$ from $w_j$ to $w_i$.
		If this were not the case, $H_i$ and $H_j$ would not be maximal, as $D_i$ would be an unlabeled local funel containing $H_i$ and $H_j$.
	Let $G_i \Subgraph D$ be a subgraph containing $P$, two incoming neighbors of $w_j$ and two outgoing neighbors of $w_i$.
	Clearly, $\InputSolutionSet$ contains some arc of $G_i$.
	Since $H_j$ and $H_i$ are local funnels, $w_j$ can only reach one sink of $H_j$, namely $u$, and $w_i$ can be reached by only one source of $H_i$, namely $s_i$.
	This means in particular that $P$ is the only path from $w_j$ to $w_i$.
	Hence, if there is any other $H_\ell \in \mathcal{F}$ that falls into this case, then the corresponding $G_\ell$ constructed is arc-disjoint to $G_i$.
	As there can be at most $k$ arc-disjoint forbidden subgraphs for funnels in $D$, there are at most $k$ local funnels in $\mathcal{F}$ that fall into this case.
		
	\Case $D_i$ is not a DAG.		
		Then there is some cycle $C_i$ containing some $w_i \in V(H_i)$ and some $w_j \in V(H_j)$.
	Clearly, $\InputSolutionSet$ contains some arc of $C_i$.
		Assume there is some other $H_\ell \in \mathcal{F}$ such that $\InN[D]{s_\ell} \cap V(H_j) \neq \emptyset$ and $D_\ell = \InducedSubgraph{D}{V(H_j) \cup V(H_\ell)}$ is not a DAG.
	Let $C_\ell$ be a cycle in $D_\ell$.
	From \cref{obs:vertex disjoint cycles local funnel} we know $C_i$ and $C_\ell$ are arc disjoint.
	As we need one arc in $\InputSolutionSet$ for each such cycle, we get that there are at most $k$ local funnels falling into this case.
	\end{CaseDistinction}
	
	Now assume $u$ is not in any local funnel in $\mathcal{H}$.
	We have two cases.
	\begin{CaseDistinction}
	\Case $u \in \Domain(\InputLabel)$.
		Then $\InputLabel(u) = \MergeL/$, as \SetLabel/ is not applicable to $s_i$.
		Since there is some $v_i \in V(H_i)$ with $\Out{v_i} > 1$ and $s_i$ can reach $v_i$, we have that $u$ can also reach $v_i$ and so $(u,s_i) \in \InputSolutionSet$ or some arc of $H_i$ is in $\InputSolutionSet$.
		Hence, there are at most $k$ local funnels $H_j \in \mathcal{F}$ with $s_j \in \OutN{u}$.
	\Case $u \not\in \Domain(\InputLabel)$.
		As $u$ is not in a local funnel, we have $\In{u} > 1$ and $\Out{u} > 1$.
		Since \SetLabel/ is not applicable, $\Out{u} \leq k$.
		Hence, there can be at most $k$ local funnels $H_j \in \mathcal{F}$ with $u \in \InN{s_j}$.
		Because \LowerBound/ is not applicable, we know there are at most $2k$ vertices $u'$ with $\In{u'} > 1$ and $\Out{u'} > 1$.
		Thus, there can be at most $2k^2$ local funnels $H_j \in \mathcal{F}$ that fall into this case.
	\end{CaseDistinction}

	By adding the bounds obtained in each case, we get $\mathcal{F} \leq k + k + k + 2k^2 \in \Bo(k^2)$. 
\end{proof}

\begin{lemma}
	\label{lemma:bound |M|}
	If there is a solution $(\InputSolutionSet, \InputSolutionLabel)$ for $(D, \InputLabel, k)$, then $\Size{\mathcal{M}} \leq k^2 + 2k$.
\end{lemma}
\begin{proof}
	Let $H_i \in \mathcal{M}$.
	\begin{CaseDistinction}
	\Case $\forall u \in \InN{s_i} : u \in \Domain(\InputLabel)$.
		As \SetLabel/ is not applicable, there is some $u \in \InN{s_i}$ with $\InputLabel(u) = \MergeL/$ and $\Out{u} > 1$.
		Hence, $\InputSolutionSet$ contains some outgoing arc of $u$.
		Any additional $H_j \in \mathcal{M}$ that falls into this case increases the outdegree of some $u'$ with $\InputLabel(u') = \MergeL/$ and $\Out{u'} > 1$.
		Thus, if there are more than $k$ local funnels $H_j \in \mathcal{M}$ that fall into this case, then $\Size{\InputSolutionSet} > k$.
	\Case There is some $u \in \InN{s_i}$ and some $H_j \in \mathcal{H}$ such that $u \in V(H_j)$.
		As $\In{s_i} > 1$ and $H_i$ is maximal, we have that $D_i = \InducedSubgraph{D}{V(H_j) \cup V(H_i)}$ is not a DAG.
		Let $C_i$ be the cycle in $D_i$.
		If there is some other $H_\ell \in \mathcal{M}$ that falls into this case, we know from \cref{obs:vertex disjoint cycles local funnel} that the corresponding cycle $C_\ell$ and $C_i$ are arc disjoint.
		As $\InputSolutionSet$ must contain one arc of each $C_\ell$, if there are more than $k$ local funnels $H_\ell$ falling into this case, then $\Size{\InputSolutionSet} > k$.
	\Case There is some $u \in \InN{s_i}$ such that $u$ is not in any local funnel and $u \not\in \Domain \InputLabel$.
		Then $\In{u} > 1$ and $\Out{u} > 1$.
		As \LowerBound/ is not applicable, there can be at most $2k$ such vertices $u$ in $D$.
		Further, $\Out{u} \leq k$ as \SetLabel/ is not applicable.
		Hence, there can be at most $2k^2$ local funnels $H_j \in \mathcal{M}$ that fall into this case.
	\end{CaseDistinction}

	By adding all cases together we obtain $\Size{\mathcal{M}} \leq + k + k + 2k^2 \in \Bo(k^2)$, as desired.
\end{proof}

From \cref{lemma:bound |X|,lemma:bound |F|,lemma:bound |M|}, we easily obtain a bound for the number of vertices in unlabeled local funnels.
Together with the fact that \LowerBound/ is not applicable, we obtain a bound for the number of unlabeled vertices in $D$.
\begin{lemma}
	\label{lemma:few unlabeled vertices}
	Let $D$ be a reduced digraph.
	Then there are $\Bo(k^5)$ vertices $v \in V(D)$ with $v \not\in \Domain(\InputLabel)$ and $\In{v} = 1 \lor \Out{v} = 1$.
\end{lemma}
\begin{proof}
	Let $\mathcal{H}$ be a maximal set of maximal vertex-disjoint local funnels in $D$.
	Clearly, every vertex $v \in V(D)$ with $v \not\in \Domain(\InputLabel)$ and $\In{v} = 1 \lor \Out{v} = 1$ is in some local funnel.
	From \cref{lemma:bound |X|,lemma:bound |F|,lemma:bound |M|} we know that $\Size{\mathcal{H}} \leq \Size{F} + \Size{M} + \Size{X} \in \Bo(k^2)$.
	Due to \cref{lemma:small local funnel}, each local funnel has at most $k^3 + 3k^2 + 2k \in \Bo(k^3)$ many vertices.
	Hence, there are at most $(5k^2 + 5k)(k^3 + 3k^2 + 2k) \in \Bo(k^5)$ vertices $v$ lying in some unlabeled local funnel.
\end{proof}

\subsection{Bounding the number of labeled vertices}
\label{subsec:labeled}

In \cref{subsec:unlabeled} we exploited the property that unlabeled vertices have bounded degree, and that we can label them if their neighborhood has some special structure captured by the reduction rules.
For the labeled vertices, however, we can apply neither of those strategies.
Instead, we first exploit the fact that we know the label of a vertex and use this to decide if an arc is never in an optimal solution or if it is always in an optimal solution.

Arcs from \MergeL/ to \ForkL/ vertices clearly need to be removed.
We show that we can also ignore arcs from \ForkL/ to \MergeL/ vertices, that is, we can remove them without changing $k$.
\begin{rrule}[\RemoveArcsName/]
	\label{rrule:remove labeled arcs}
	Let $(v,u) \in A(D)$.
	If $\Label(v) = \ForkL/$ and $\Label(u) = \MergeL/$, remove $(v,u)$.
	If $\InputLabel(v) = \MergeL/$ and $\InputLabel(u) = \ForkL/$, remove $(v,u)$ and decrease $k$ by 1. 
\end{rrule}
\begin{proof}[Proof of safety of \RemoveArcs/]
	Clearly, a solution for $(D,\InputLabel, k)$ is also a solution for the reduced instance $(D', \ReducedLabel, k')$.
	So let $(\ReducedSolutionLabel, \ReducedSolutionSet)$ be a solution for the reduced instance.

	If $\InputLabel(v) = \MergeL/$ and $\InputLabel(u) = \ForkL/$, then $(v,u)$ is in any solution.
	Hence, $(\ReducedSolutionLabel, \ReducedSolutionSet \cup \{(v,u)\})$ is a solution for the input instance and $\Size{\ReducedSolutionSet \cup \{(v,u)\}} \leq k$.

	If $\InputLabel(v) = \ForkL/$ and $\InputLabel(u) = \MergeL/$, then we claim $\ReducedSolutionLabel$ is a funnel labeling for $D - \ReducedSolutionSet$.
	If $D - \ReducedSolutionSet$ is a DAG, then the claim trivially holds.

	Now assume towards a contradiction that $D - \ReducedSolutionSet$ is not a DAG.
	Then there is a cycle $C$ using the arc $(v,u)$.
	This implies that there is a path $P$ from $u$ to $v$ in $D - \ReducedSolutionSet$.
	In particular, this path also exists in $D' - \ReducedSolutionSet$ since it does not use the arc $(v,u)$.
	However, as $\InputLabel(u) = \MergeL/$ and $\InputLabel(v) = \ForkL/$, we know there is some arc $(v', u')$ in $P$ with $\ReducedSolutionLabel(v') = \MergeL/$ and $\ReducedSolutionLabel(u') = \ForkL/$.
	But then $\ReducedSolutionLabel$ is not a funnel labeling for $D' - \ReducedSolutionSet$, a contradiction.
	Hence, a solution for the reduced instance implies a solution for the input instance, proving the rule is safe.
\end{proof}

We now identify certain vertices that can be removed safely.
Clearly, sources and sinks cannot be in any cycle in $D$.
By carefully considering the neighborhood of a source or sink $v$, we can also prove that $v$ is not ``relevant'' for any forbidden subgraph for funnels in $D$.
\begin{rrule}[\RemoveSinksName/]
	\label{rrule:remove labeled sources and sinks}
	Let $v \in V(D)$ be a labeled vertex where $\OutN{v} \cup \InN{v} \subseteq \Domain(\InputLabel)$.
	Remove $v$ if one of the following holds.
	\begin{enumerate}
		\item $\In{v} = 0$ and no $u \in \OutN{v}$ exists with $\InputLabel(u) = \ForkL/$ and $\In{u} > 1$, or
		\item $\Out{v} = 0$ and no $u \in \InN{v}$ exists with $\InputLabel(u) = \MergeL/$ and $\Out{u} > 1$.
	\end{enumerate}
\end{rrule}
\begin{proof}[Proof of safety of \RemoveSinks/]
	Let $(D', \ReducedLabel, k)$ be the reduced instance.
	Clearly, if $(\InputSolutionSet, \InputSolutionLabel)$ is a solution for $(D, \InputLabel, k)$, then, after restricting $(\InputSolutionSet, \InputSolutionLabel)$ to the vertices in $D'$, we also have a solution for $(D', \ReducedLabel, k)$.

	Now let $(\ReducedSolutionSet, \ReducedSolutionLabel)$ be a solution for the reduced instance.
	We consider the case where $\In[D]{v} = 0$ and there is no $u \in \OutN{v}$ such that $\InputLabel(u) = \ForkL/$ and $\In{u} > 1$.
	The other case follows analogously.
	We claim that the labeling $\InputSolutionLabel \supseteq \ReducedSolutionLabel$ with $\InputSolutionLabel(v) = \InputLabel(v)$ is a funnel labeling for $D - \ReducedSolutionSet$.
	Since $\In[D - \ReducedSolutionSet]{v} = 0$, there can be no cycle containing $v$.
	There can be no vertex $u \in V(D)$ with $\ReducedSolutionLabel(u) = \MergeL/$ and $\Out[D - \ReducedSolutionSet]{u} > 1$ as this would imply $\Out[D' - \ReducedSolutionSet]{u} > 1$ because $\In[D - \ReducedSolutionSet]{v} = 0$, a contradiction.
	There can also be no vertex $u \in V(D)$ with $\ReducedSolutionLabel(u) = \ForkL/$ and $\In[D - \ReducedSolutionSet]{u} > 1$ as all $u \in \OutN[D - \ReducedSolutionSet]{v}$ with $\InputLabel(u) = \ForkL/$ have $\In[D - \ReducedSolutionSet]{u} = 1$, and $\OutN[D]{v} \subseteq \Domain(\InputLabel)$.

	Because \RemoveArcs/ is not applicable, if $\InputLabel(v) = \MergeL/$, there is no $u \in \OutN[D - \ReducedSolutionSet]{v}$ with $\InputLabel(u) = \ForkL/$.
	Hence, there is no arc $(w,x) \in A(D - \ReducedSolutionSet)$ with $\ReducedSolutionLabel(w) = \MergeL/$ and $\ReducedSolutionLabel(x) = \ForkL/$.
	From the labeling characterization from \cref{thm:funnel characterization} we have that $(\ReducedSolutionSet, \ReducedSolutionLabel)$ is solution for $(D, \InputLabel, k)$ and the reduction rule is safe.
\end{proof}

Having exhaustively applied \cref{rrule:remove labeled arcs,rrule:remove labeled sources and sinks}, we can bound the number of labeled vertices in $D$.
Since \LowerBound/ is not applicable, we already have a bound for the number of vertices $v$ with $\InputLabel(v) = \ForkL/ \land \In{v} > 1$ or $\InputLabel(v) = \MergeL/ \land \Out{v} > 1$.
Hence, we only need to consider vertices in the set $L = \{v \in \Domain(\InputLabel) \mid \InputLabel(v) = \ForkL/ \land \In{v} \leq 1 \text{ or } \InputLabel(v) = \MergeL/ \land \Out{v} \leq 1\}$.

To bound $\Size{L}$, we exploit the bound on the number of unlabeled vertices from \cref{lemma:few unlabeled vertices} and also the fact that such vertices have small degree as \SetLabel/ is not applicable.
We first partition $L$ into two subsets	$L_1 = \{v \in L \mid \InN{v} \cup \OutN{v} \not\subseteq \Domain(\InputLabel)\}$ and $L_2 = L \setminus L_1$.

\begin{lemma}
	\label{lemma:L1 small}
	$\Size{L_1} \in \Bo(k^6)$.
\end{lemma}
\begin{proof}
	Let $U$ be the set of unlabeled vertices.
	Clearly $L_1 \subseteq \InN{U} \cup \OutN{U}$.
	As \SetLabel/ is not applicable, we have $\In{v} \leq k+1$ and $\Out{v} \leq k+1$ for every $v \in U$.
	From \cref{lemma:few unlabeled vertices} we know $\Size{U} \in \Bo(k^5)$.
	Hence, $\Size{L_1} \leq \Size{\InN{U} \cup \OutN{U}} \in \Bo(k^6)$.
\end{proof}

\begin{lemma}
	\label{lemma:L2 small}
	$\Size{L_2} \in \Bo(k)$.
\end{lemma}
\begin{proof}
	Let $V_\ForkL/ = \{v \mid \InputLabel(v) = \ForkL/\}$ and	$L_\ForkL/ = V_\ForkL/ \cap L_2$.
	The case for the vertices labeled with \MergeL/ follows analogously.

	Since \RemoveArcs/ is not applicable, we have $\InputLabel(u) = \ForkL/$ for all $u \in \OutN{L_\ForkL/} \cup \InN{L_\ForkL/}$.
	Let
	
		$R_1 = \{u \in V_\ForkL/ \mid \In{u} > 1\}$,
		$R_2 = \{u \in L_\ForkL/ \mid \In{u} \leq 1, \OutN{u} \cap R_1 \neq \emptyset\}$ and
		$R_3 = \{u \in L_\ForkL/ \mid \In{u} \leq 1, \OutN{u} \cap R_1 = \emptyset\}$.
	
	Note that $L_2 = R_2 \cup R_3$ and $R_1 \cap L_2 = \emptyset$.

	A solution set $\InputSolutionSet \subseteq A(D)$ must contain at least $\In{v} - 1$ many incoming arcs of $v$ for every $v \in R_1$.
	As each $u \in R_2$ has some $v \in R_1$ as outneighbor, we have $\Size{R_2} \leq 2k$.

	Let $v \in R_3$.
	We claim that $v$ can reach some vertex of $R_2$.
	Since \RemoveSinks/ is not applicable and $\OutN{v} \cap R_1 = \emptyset$, we have $\In{v} = 1$ and $\Out{v} \geq 1$.
	This means that, if we successively follow the outneighbors of $v$, we reach a vertex of $R_2$ or find a cycle $C$ using only vertices of $R_3$.
	However, as \BreakCycle/ is not applicable, such a cycle $C$ cannot exist:
	every vertex $v \in R_3$ has $\In{v} = 1$ and $\InputLabel(v) = \ForkL/$, implying we could apply \BreakCycle/ to $C$.
	Hence, every vertex of $R_3$ can reach some $u \in R_2$.

	We greedily construct vertex-disjoint paths $P_1, P_2, \dots, P_{a}$ ending in $R_2$ whose inner vertices lie in $R_3$.
	For a vertex $v \in R_3$ take an arbitrary $u \in R_2$ such that $v$ can reach $u$.
	Consider a path $P$ from $v$ to $u$.
	If none of its vertices lie in any already constructed $P_i$, we just take the path $P$ into our set of paths.
	Otherwise, assume that $P$ intersects some $P_i$ at $w$ and let $w$ be the first such vertex in $P$.
	Since the indegree of any vertex in $R_3 \cup R_2$ is at most one, we know that $w$ is the starting point of $P_i$.
	Hence, we can obtain a path $P_j$ by taking the path from $v$ to $w$ in $P$ and then concatenating $P_i$.
	As $w$ is the first vertex of $P$ intersecting any other path, we get that $P_j$ only intersects $P_i$.
	By replacing $P_i$ with $P_j$, we obtain a path that also contains $v$.
	We repeat this process until we covered all $v \in R_3$.

	Since $\Size{R_2} \leq 2k$, we have $a \leq 2k$.
	We now prove that $\Size{V(P_i)} \leq 4$ for any $P_i$ in our set of vertex-disjoint paths.
	Note that $\In{u} \leq 1$ for any vertex $u \in V(P_i)$.

	Since \Dissolve/ is not applicable, any inner vertex $u$ of $P_i$ has $\Out{u} > 1$.
	Let $w$ be the successor of $u$ in $P_i$.
	As \ShiftNeighbors/ is not applicable, we have that $\In{w} > 1$ or $v \in \OutN{u}$ where $v$ is the unique inneighbor of $u$.
	If $\In{w} > 1$, then $u \in R_2$ and is the endpoint of $P_i$, a contradiction to the assumption that $u$ is an inner vertex of $P_i$.
	Otherwise, we know that $u \not\in \OutN{w}$ as $\In{u} = 1$.
	If $u$ is the only inner vertex of $P_i$, then $\Size{V(P_i)} \leq 3$.
	Otherwise, its successor $w$ in $P_i$ is an inner vertex of $P_i$ (since $v$ is the starting point of $P_i$, and so $v \not\in \OutN{u}$).
	Hence, we can apply the same argumentation to $w$ and conclude that it has some outneighbor $x$ with $\In{x} > 1$, implying $x \in R_2$ and $\Size{V(P_i)} \leq 4$.

	Since $\Size{V(P_i)} \leq 4$ and $a \leq 2k$, we have that $\Size{L_2} \leq 6k$.
	Because $L_2 = R_2 \cup R_3$, we have that $\Size{L_2} \leq 8k \in \Bo(k)$, as desired.
\end{proof}

\begin{lemma}
	\label{lemma:D is small}
	Let $(D, \InputLabel, k)$ be an \FADL/ instance where \cref{rrule:lower bound,rrule:set label,rrule:dissolve vertex,rrule:break cycle,rrule:shift neighbours,rrule:local funnel labeled neighbor,rrule:remove labeled arcs,rrule:remove labeled sources and sinks} are not applicable.
	Then $\Size{V(D)} \in \Bo(k^6)$ and $\Size{A(D)} \in \Bo(k^6)$.
\end{lemma}
\begin{proof}
	As \LowerBound/ is not applicable, there are at most $2k$ vertices $v$ with $\In{v} > 1$ and $\Out{v} > 1$, and also at most $2k$ many vertices $v$ with $\InputLabel(v) = \ForkL/ \land \In{v} > 1$ or $\InputLabel(v) = \MergeL/ \land \Out{v} > 1$.
	From \cref{lemma:few unlabeled vertices} we know there are $\Bo(k^5)$ many unlabeled vertices $v \in V(D)$ with $\In{v} \leq 1$ or $\Out{v} \leq 1$.
	Finally, due to \cref{lemma:L1 small,lemma:L2 small} there are $\Bo(k^6)$ vertices $v$ with $\InputLabel(v) = \ForkL/ \land \In{v} \leq 1$ or $\InputLabel(v) = \MergeL/ \land \Out{v} \leq 1$.
	As any vertex in $D$ falls into one of these groups, we have $\Size{V(D)} \in \Bo(k^6)$.

	As \RemoveArcs/ is not applicable, there is no arc $(v,u)$ where $v,u \in \Domain(\InputLabel)$ and $\InputLabel(v) \neq \InputLabel(u)$.
	Since there are $\Bo(k^5)$ many unlabeled vertices and every unlabeled vertex has in- and outdegree at most $k+1$, there are $\Bo(k^6)$ arcs $(v,u)$ where $v \not\in \Domain(\InputLabel)$ or $u \not\in \Domain(\InputLabel)$.

	Now let $(v,u)$ be some arc where $v,u \in \Domain(\InputLabel)$.
	Note that $\InputLabel(v) = \InputLabel(u)$.
	\begin{CaseDistinction}
		\Case $v,u \in L$.
		Then $\Out{v} = 1$ (if $\InputLabel(v) = \MergeL/$) or $\In{u} = 1$ (if $\InputLabel(u) = \ForkL/$).
		Thus, there can be at most $\Size{L} \in \Bo(k^6)$ many arcs $(v,u)$ where $v,u \in L$.
		\Case $v,u \not \in L$.
		As \LowerBound/ is not applicable, there can be at most $2k$ such vertices.
		Thus, there are at most $4k^2$ arcs between labeled vertices not in $L$.
		\Case Exactly one of $v,u$ is in $L$.
		\begin{Subcase}
			\Case $v \not\in L \land \InputLabel(v) = \ForkL/$ or $u \not\in L \land \InputLabel(u) = \MergeL/$.
			Then $\In{u} = 1$ or $\Out{v} = 1$.
			Hence, there can be at most $\Size{L} \in \Bo(k^6)$ such arcs.
			\Case $v \not\in L \land \InputLabel(v) = \MergeL/$ or $u \not\in L \land \InputLabel(u) = \ForkL/$.
			If $v \not\in L$, then at least half of its outgoing arcs need to be in a solution set.
			Similarly, if $u \not\in L$, at least half of its incoming arcs need to be in a solution set.
			Hence, there can be at most $2k$ many arcs falling into this case.
		\end{Subcase}
	\end{CaseDistinction} 
	By adding all cases together, we obtain that $\Size{A(D)} \in \Bo(k^6)$, concluding the proof.
\end{proof}

\section{Computing the Kernel}
\label{sec:compute}

In \cref{subsec:unlabeled,subsec:labeled} we defined the reduction rules for the kernelization process and showed that, if none of the reduction rules are applicable to a digraph $D$, then the size of $D$ is polynomially bounded on $k$.
To conclude the proof that \FADS/ admits a polynomial problem kernel, we show that it is possible to apply all reduction rules in $\Bo(nm)$ time and also reduce the \FADL/ instance back into an \FADS/ instance.
\begin{lemma}
	\label{lemma:kernel time}
	We can exhaustively apply \cref{rrule:lower bound,rrule:set label,rrule:dissolve vertex,rrule:break cycle,rrule:shift neighbours,rrule:local funnel labeled neighbor,rrule:remove labeled arcs,rrule:remove labeled sources and sinks} in $\Bo(nm)$ time to an \FADL/ instance $(D, \InputLabel, k)$, where $n = \Size{V(D)}$ and $m = \Size{A(D)}$.
\end{lemma}
\begin{proof}
	We apply the reduction rules exhaustively in the order they are defined.
	In order to do so efficiently, we use a constant number of counters for each vertex $v$ in order to check if a reduction rule is applicable to $v$.

	To apply \LowerBound/, \SetLabel/ and \Dissolve/ we only need to check the labels and degrees of a vertex and its neighbors.
	Whenever the label of a vertex changes, we need to recheck if we can apply the reduction rules to its neighbors.
	It is therefore sufficient to store the degree of $v$ and the number of its neighbors which have the a certain type (for example, the number of $u \in \OutN{v}$ with $\InputLabel(u) = \ForkL/$ and $\In{u} = 1$).
	Whenever the label of a vertex changes, we only need to increment the counters its neighbors.
	As we set the label of a vertex at most once, we need to visit each arc constantly many times.

	For \ShiftNeighbors/, we consider the first case of the rule where the indegrees are one (the other case is applied analogously).
	We search for a vertex $w$ with $\In{w} = 1$.
	We then take the unique inneighbor $v \in \InN{w}$, tracking in a counter the number of $w \in \OutN{v}$ with $\In{w} = 1$.
	If $\In{v} = 1$, we construct a path $P$ ending in $v$ by following its unique inneighbor until we obtain a vertex $y$ with $\In{y} > 1$ or $\InputLabel(y) = \MergeL/$, which we do not include in $P$.
	We denote the starting point of $P$ by $u$.
	
	If $u = w$, we can apply \BreakCycle/ by deleting the arc $(v,w)$ and labeling every vertex in $P$ with $\ForkL/$ (if they are not already labeled).
	Now assume $P$ is indeed a path.
	For each $x \in \OutN{P}$ we count the number $c = \Size{\In{x} \cap V(P)}$ and then shift all arcs $(v,x)$ with $v \in V(P)$ by arcs coming from the first $c$ vertices in $P$.
	The cost of applying this operation is linear on the number of arcs leaving $P$.
	If a vertex $v$ is in another such path $Q$, then after applying \ShiftNeighbors/ once to $v$ we shift $Q$ in such a way that only the startpoint $u$ is in more than one path to which the rule is applicable.
	
	We only need to recheck if \ShiftNeighbors/ is applicable to $v$ when an arc is removed from it by \RemoveArcs/ or by \BreakCycle/ (arcs removed from \RemoveSinks/ never trigger \ShiftNeighbors/ due to the degrees of the affected vertices).
	In this case, we apply depth-first search on $v$, following its outgoing arcs, in order to find all paths to which \ShiftNeighbors/ is applicable.
	As observed above, the computational cost of shifting the arcs of a path $P$ is linear in the number of arcs leaving $P$.
	Since no reduction rule adds arcs to $D$, we only need to recheck if \ShiftNeighbors/ is applicable to $v$ once.
	Hence, each arc gets shifted $\Bo(n)$ many times, and this rule can be exhaustively applied in $\Bo(nm)$ time.

	\RemoveArcs/ is applied whenever we set the label of a vertex $v$.
	In order to apply it, it suffices to iterate through $\InN{v}$ or $\OutN{v}$ and check the labels of those vertices.
	Since we set the label of a vertex at most once, we need $\Bo(n + m)$ time in total for this rule.
	Finally, we apply \RemoveSinks/ by applying breadth-first search on the sources and sinks of $D$.
	Note that, although \RemoveSinks/ remove vertices (and arcs) from $D$, it cannot cause \BreakCycle/, \ShiftNeighbors/ or other reduction rules to become applicable if they were not applicable before.
	Hence, by applying \RemoveSinks/ only after no other rule is applicable, we can exhaustively apply this rule in $\Bo(n + m)$ time, giving us a total running time of $\Bo(nm)$ for the kernelization process, concluding the proof.
\end{proof}
\begin{theorem}
	\label{thm:kernel}
	\FADS/ admits a kernel with $\Bo(k^6)$ vertices and $\Bo(k^7)$ arcs which can be computed in $\Bo(nm)$ time, where $n = \Size{V(D)}$, $m = \Size{A(D)}$ and $D$ is the input digraph.
\end{theorem}
\begin{proof}
	We start by reducing the \FADS/ instance into an \FADL/ instance $(D, \InputLabel, k)$ by adding an empty labeling $\InputLabel$.
	Using \cref{lemma:kernel time}, we can exhaustively apply all reduction rules to $(D, \InputLabel, k)$ in $\Bo(nm)$ time.
	
	From \cref{lemma:D is small} we know $\Size{V(D)} \in \Bo(k^6)$ and $\Size{A(D)} \in \Bo(k^6)$.
	We now reduce the \FADL/ instance back into an \FADS/ instance $(D', k)$ in order to obtain a kernel for the original problem.
	We first set $D' \Set D$ and add $k+2$ vertices $f_1, f_2, \dots, f_{k+2}$ and $k+2$ vertices $m_1, m_2, \dots, m_{k+2}$ to $D'$.
	Let $v \in \Domain(\InputLabel)$.
	If $\InputLabel(v) = \ForkL/$, we add the arc $(v, f_i)$ for each $1 \leq i \leq k+2$.
	If $\InputLabel(v) = \MergeL/$, we add the arc $(m_i, v)$ for each $1 \leq i \leq k+2$.

	Trivially, a solution for the \FADL/ instance is also a solution for the \FADS/ instance.
	It is also easy to see that, if there is some arc set $\ReducedSolutionSet \subseteq A(D')$ and some funnel labeling $\ReducedSolutionLabel$ for $D' - \ReducedSolutionSet$ such that $\InputLabel(v) \neq \ReducedSolutionLabel(v)$ for some $v \in \Domain(\InputLabel)$, then $\Size{\ReducedSolutionSet} > k$.
	Hence, a solution for $(D', k)$ implies a solution for $(D, \InputLabel, k)$.

	We added $2k+4$ vertices and $\Bo(k^7)$ many arcs to $D'$, and so $\Size{V(D')} \in \Bo(k^6)$ and $\Size{A(D')} \in \Bo(k^7)$, thus concluding the proof.
\end{proof}

\section{Conclusion}
The kernelization algorithm provided in this paper heavily relies on the characterizations of \cref{thm:funnel characterization} for funnels.
Both the characterization by forbidden subgraphs as well as the labeling characterization allowed us to derive reduction rules based only on ``local'' substructures as the degree or neighborhood of a vertex.
In a sense, this ``locality'' property saved us from computing any set of vertex-disjoint local funnels, despite the fact that the results and reduction rules from \cref{subsec:unlabeled} heavily rely on local funnels.

The polynomial kernels for \CVDS{\ProblemName{Out-Forest}} and \CVDS{\ProblemName{Pumpkin}} due to \cite{mnich2017polynomial} also rely on ``localized'' forbidden substructures.
We consider that generalizing these results to larger digraph classes of unbounded treewidth, but which are characterized by forbidden substructures, to be a very interesting direction for future research.

Further, it would also be interesting to decide if \CVDS{\ProblemName{Funnel}} admits a polynomial kernel or not (it is in FPT with respect to the solution size \cite{Millani18}), especially since a kernel for this problem would require considerably different ideas from the ones presented in this paper, as it is no longer clear how to exploit the vertex labeling in the vertex-deletion setting.

\bibliography{references.bib}

\appendix
\ifappendix

\fi
\end{document}